\documentclass{IEEEoj}
\usepackage{cite}
\usepackage{amsmath,amssymb,amsfonts}
\usepackage{algorithmic}
\usepackage{algorithm}

\newtheorem{thm}{Theorem}

\def\BibTeX{{\rm B\kern-.05em{\sc i\kern-.025em b}\kern-.08em
    T\kern-.1667em\lower.7ex\hbox{E}\kern-.125emX}}
\AtBeginDocument{\definecolor{ojcolor}{cmyk}{0.93,0.59,0.15,0.02}}
\def\OJlogo{\vspace{-14pt}\includegraphics[height=24pt]{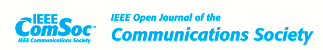}}
\usepackage{bm}
\usepackage{comment}
\usepackage[hang,small,bf]{caption}
\usepackage[subrefformat=parens]{subcaption}
\usepackage{latexsym}
\usepackage[varg]{newtxmath}
\usepackage[whole]{bxcjkjatype}
\usepackage{graphicx,color}
\usepackage{textcomp}
\begin{document}
\receiveddate{XX Month, XXXX}
\reviseddate{XX Month, XXXX}
\accepteddate{XX Month, XXXX}
\publisheddate{XX Month, XXXX}
\currentdate{XX Month, XXXX}
\doiinfo{OJCOMS.2022.1234567}

\title{AP Connection Method for Maximizing Throughput Considering User Moving and Degree of Interference Based on Potential Game}

\author{Yu Kato\authorrefmark{1} (MEMBER, IEEE), Jiquan Xie\authorrefmark{2} (MEMBER, IEEE), Tutomu Murase\authorrefmark{2} (MEMBER, IEEE), AND Sumiko Miyata\authorrefmark{1} (MEMBER, IEEE)}
\affil{Shibaura Institute of Technology, 3-7-5 Toyosu, Koto-ku, Tokyo, 135-8548, Japan}
\affil{Nagoya University, Furocho, Nagoya Chikusa-ku, Aichi, Tokyo, 464--8601, Japan}
\corresp{CORRESPONDING AUTHOR: Yu Kato (e-mail: ma22035@sic.shibaura-it.ac.jp).}
\authornote{This work was supported by the National Institute of Information and Communications Technology (NICT) and the JSPS KAKENHI Grant Numbers JP19K11947, JP22K12015, JP20H00592, JP21H03424.}
\markboth{Preparation of Papers for IEEE OPEN JOURNALS}{Author \textit{et al.}}

\begin{abstract}
For multi-transmission rate environments, access point (AP) connection methods have been proposed for maximizing system throughput, which is the throughput of an entire system, on the basis of the cooperative behavior of users. These methods derive optimal positions for the cooperative behavior of users, which means that new users move to improve the system throughput when connecting to an AP. However, the conventional method only considers the transmission rate of new users and does not consider existing users, even though it is necessary to consider the transmission rate of all users to improve system throughput. In addition, these method do not take into account the frequency of interference between users. 

In this paper, we propose an AP connection method which maximizes system throughput by considering the interference between users and the initial position of all users. In addition, our proposed method can improve system throughput by about \(6 \, \%\) at most compared to conventional methods.
\end{abstract}



\maketitle

\section{INTRODUCTION}
\IEEEPARstart{I}{n} recent years, advances in communication technology have increased the number of users of wireless terminal devices such as smartphones, tablets, and personal computers. In addition, public wireless LAN services have been widely deployed in cafes, offices, schools, and other public facilities due to the rapid spread of wireless LANs. In this environment, multiple users share a same access point (AP). 

In general, total throughput of an AP (System throughput) is reduced when one user requests an extremely low transmission rate compared with  other users in a multi transmission rate environment where each user has a different transmission rate \cite{1208921}. Therefore, AP connection methods have been proposed for selecting an appropriate position in a multi-transmission rate environment in order to solve this problem \cite{1208921}\cite{miyata_CCNC}\cite{kato_CCNC}.

In conventional work, an AP connection method has been proposed to maximize system throughput based on users' cooperative behavior \cite{miyata_CCNC}. The user's cooperative behavior means that the user moves to improve the system throughput. Therefore, this conventional work assumes that the user is moving, and this model changes the transmission rate depending on the user's position. However, this conventional work did not consider interference, even though users need to move to positions where interference is considered in order to improve system throughput. Here, interference means packet collisions between users. 

On the other hand, other conventional work has proposed an AP connection method that maximizes system throughput based on users' cooperative behavior and considering the interference between users \cite{kato_CCNC}. However, this work assumes that there is always interference between users. In recent wireless LAN environments that are equipped with CSMA/CA, it is difficult to imagine a situation where there is always interference between users because CSMA/CA has the ability to control packet collisions \cite{csma_ca_collision}. Therefore, for a realistic analysis, it is necessary to consider the degree of interference between users.

In addition, there are other problems with these works \cite{miyata_CCNC}\cite{kato_CCNC}. To improve system throughput, the transmission rate of all users and interference need to be considered. However, the conventional works only consider the transmission rates of a limited user  and do not consider all users. In general, there are users who are fixed in conventional works, although there is a probability that some fixed user may be moving \cite{user_move}. We need to consider the performance of AP connection method assuming some fixed users also move.

In this paper, we solve these problems by modeling them using a potential game. By using a potential game, we can find at least one optimal solution because we can make the incentives of all users a single function. In particular, we add terms of probability of no packet collision and probability of packet collision to the system throughput equation used in conventional works. Through this approach, we consider the degree of interference. In addition, we propose an AP connection method that maximizes system throughput by extending the adaptive range of potential game players to all users instead of only limited users. By extending the adaptive range of players in the potential game to all users from only limited users, we can eliminate the case where only one user's transmission rate is extremely low. As a result, we can prevent the loss of overall system performance.

From the above discussion, our contribution regards three factors:
\begin{quote}
\begin{itemize}
\item Realistic theoretical analysis by considering both cases in which there is interference and cases in which there is no interference.
\item Proposed method for maximizing system throughput by considering the initial positions of users connected to the same AP.
\item Analysis of optimal user position that maximizes system throughput by considering interference using the potential game.
\end{itemize}
\end{quote}

The structure of this paper is as follows. Section 2 give definitions as a preliminary preparation, Section 3 describes related work on AP connection method, and Section 4 discusses the proposed method in detail, including the utility function. Section 5 presents a numerical analysis, and Section 6 summarizes and discusses the problems and issues in this work.

\section{RELATED WORK}
In general, conventional research for wireless networks aims to improve system throughput. To improve throughput, there is a work field that focuses on AP connection method. We can divide the connection to an AP into two types: methods that connect to  APs without user movement, and methods that connect to APs with user movement.

Regarding methods of connecting to an AP without user movement, Raschellà et al. proposed a method for selecting the appropriate AP by focusing on each user application to increase throughput \cite{7500386}. In addition, Raschellà et al. theoretically prove the existence of a solution by using a potential game \cite{8406147}. However, this method does not consider user movement \cite{7500386}\cite{8406147}. Thus, if a new user arrives at a position far from the AP, the user must connect to the AP at the arrival position. This results in a low transmission rate, which causes performance anomalies. Moreover, Liyanage et al. proposed an appropriate AP connection method that considers interference between users connected to the AP and the distance between the AP and the users \cite{RSSI_AP_selection}. However, this method also does not consider the user's movement. Therefore, it is possible that user interference may increase and throughput may decrease depending on the user's initial position. On the other hand, Liu et al. proposed a method which connects to an AP considering transmission rates of each user in a multi-rate environment, not only RSSI \cite{RSSI_AP_selection}. However, the throughput used in this method assumes no packet collisions between users. In other words, it assumes that there is no interference between users. Thus, this method differs from the real environment in a WLAN environment where multiple users are connected.

Regarding methods of connecting to an AP with user movement, the methods of connection can be divided into two types: indirect connection between the AP and a new user via relaying terminals such as other users or drones.

Regarding indirect connection methods, Yanai et al. optimized the flight paths of drones in a wireless relay network \cite{yanai_drone}. As a result, the transmission delay time was reduced. However, this method does not consider the interference between drones. Moreover, it does not take into account the signal strength between drones. In addition, Xie et al. derived the optimal position between users in an ad hoc network in  consideration of the interference between each user \cite{9136693}. Furthermore, Anjiki et al. considered dynamic routing according to user mobility in ad hoc networks \cite{Ad_hoc_network}. As a result, appropriate user selection and user positioning were analyzed to improve throughput. However, these methods are multi-hop and assume that all nodes can move, whereas WLANs have the limitation of a fixed AP.

Regarding direct connection methods, Miyata et al. proposed a method that improves system throughput on the basis of users' cooperative behavior \cite{miyata_CCNC}\cite{miyata_IEICE}. Yamori et al. proposed a method that improves system throughput by encouraging them to select other APs by encouraging users to move \cite{yamori}. However, these methods have a  problem in that they do not consider interference between users.

\section{PRELIMINARY}
\subsection{HIDDEN TERMINAL PROBLEM}
In a wireless environment, packet collisions occur when two terminals that cannot detect each other's communications send packets to a receiving terminal.

\subsection{SHANNON'S THEOREM}
Shannon's theorem expresses the limit of transmission rate which allows information to be transmitted without error in a steady state. The Shannon's theorem is defined by \cite{Shannon}:
\begin{equation}
R_{X} = W\log_{2}{\{1 + \Gamma_{X}\}},
\label{equ_trans}
\end{equation}
where \(W\) represents the channel bandwidth. Moreover, \(\Gamma_{X}\) is the SINR (Signal to Interference and Noise power Ratio), which is the ratio of a user \(X\)'s own signal power to the interference power.

\subsection{CAPTURE EFFECT}
In general, it is known that an AP can communicate with a user \(X\) without packet errors if the signal power from a user \(X\) received by an AP is stronger than that of another user \cite{capture_effect}\cite{capture_effect_WLAN}. This phenomenon is called the capture effect and can be expressed as SINR must be larger than a threshold, which can be defined as \cite{capture_effect}\cite{capture_effect_WLAN}:
\begin{equation}
\Gamma_{X} \geq \delta.
\label{SINR_equ}
\end{equation}
Let \(\delta\) be a threshold.

\subsection{GAME THEORY}
Game theory is the study of how players such as servers or people interact and make decisions. Players select their own behaviors
to achieve their own objectives in a game. For wireless communication environments where multiple users exist and users must be considered, game theory is one of the best theories to use for analysis \cite{reason_for_game_theory_AP_selection}\cite{reason_for_game_theory_in_wireless_LAN_ver1}\cite{reason_for_game_theory_in_wireless_LAN_ver2}.

Game theory consists of three elements: players, strategies, and utility functions, which can be defined as \cite{game_theory}:
\begin{equation}
\{\mathcal{S},\mathcal{A},(u_{i})_{i \in \mathcal{S}}\}.
\label{equ_game}
\end{equation}
Let \(\mathcal{S} = \{1,...,S\}\) and let \(\mathcal{A}_{i}\) be a set of strategies \(a_{i}\) of each player \(i\). The set of strategies of all players is denoted by \(\mathcal{A}\), and \(\mathcal{A}\) denotes the direct product set as \(\mathcal{A}_{1} \times \cdots \times \mathcal{A}_{S}\). Each element \(\bm {a} = (a_1,a_2,...a_S) \in \mathcal {A}\) is said to be a strategy profile  \cite{BSI_game}. Let \(u_i({\bm a}):\mathcal A \longmapsto \mathbb{R}\) be the utility function when player \(i \in \mathcal S\) takes strategy \(a_{i}\), where \(\mathbb{R}\) is a set of real numbers. All players use this common knowledge to solve the game theory.

\begin{figure}[t]
\centering
\includegraphics[scale=0.25]{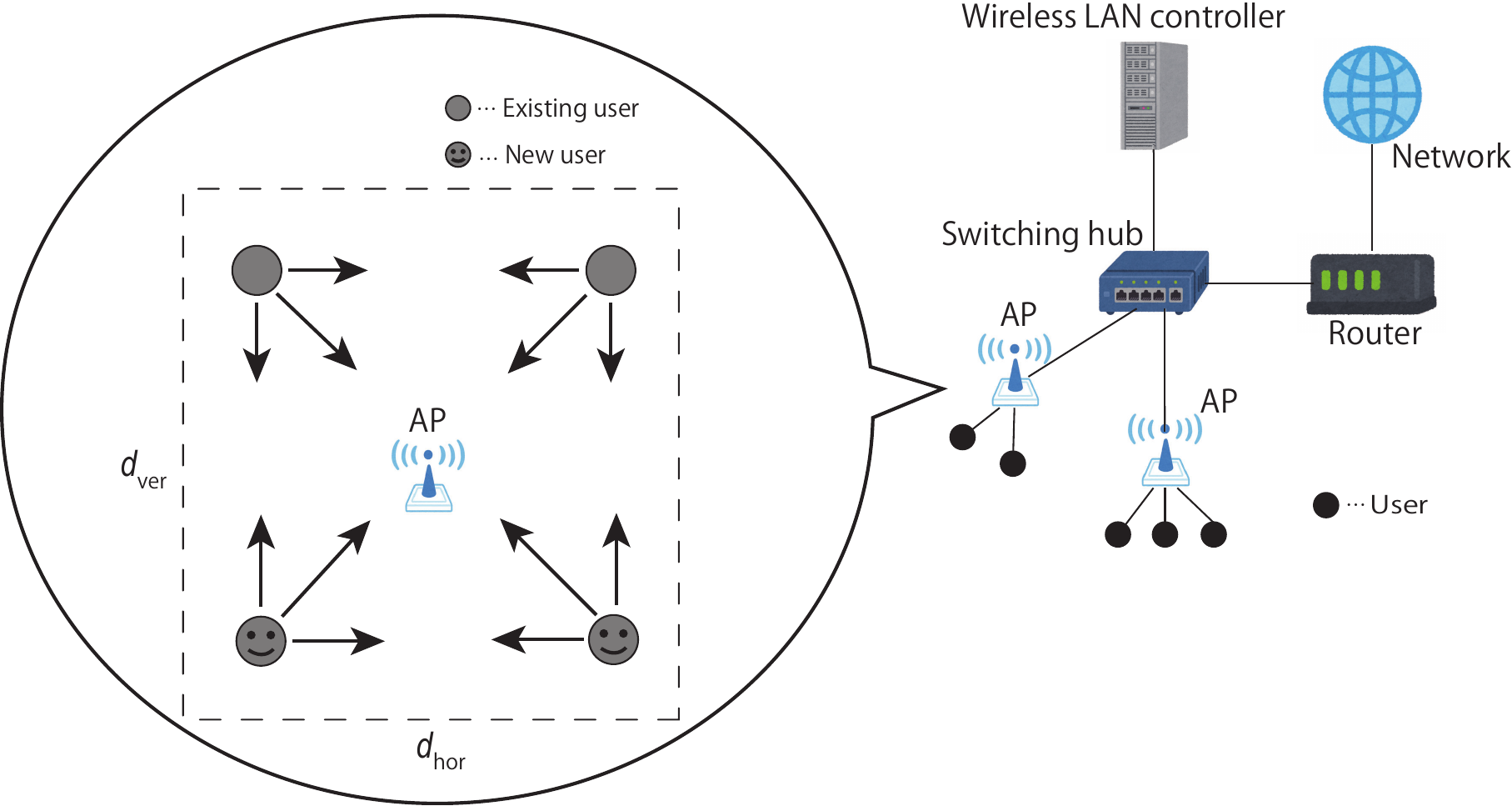}
\caption{System model.}
\label{fig1}
\end{figure}

\subsection{NASH EQUILIBRIUM}
A Nash equilibrium is a solution concept in game theory, that is, a combination of strategies in which no player can obtain a higher utility value by changing his or her strategy. The final objective of this paper is to find this Nash equilibrium.

Let \(\bm{a_{-i}} = (a_{1},...,a_{i-1},a_{i+1},...,a_{S})\) be the strategy except for player \(i\) as a game theory-specific expression. The Nash equilibrium is defined by \cite{game_theory}.
\begin{equation}
u_{i}(a^{*}_{i} , \bm{a^{*}_{-i}}) \geq u_{i}(a'_{i} , \bm{a^{*}_{-i}}), \forall a'_{i} \in \mathcal{A}, \forall i \in \mathcal{S},
\label{equ_utility_1}
\end{equation}
where \(a^{*}_{i}\) represents the strategy that is the Nash equilibrium for player \(i\).

\subsection{BSI GAME (BILATERAL SYMMETRIC INTERACTION GAME)}
A BSI game is one of a class of strategy games. In addition, a game, \(\mathscr{G}_{\mathrm {bsi}} := \{\mathcal{S},\mathcal{A},(u_{i})_{i \in \mathcal{S}}\}\), is a BSI game if its utility function satisfies the following definition \cite{BSI_game}\cite{BSI_game_ver2}.
\begin{equation}
u_{i}(a_{i} , \bm{a_{-i}}) = \sum_{j \in \mathcal {N} \backslash \{i\}} w_{i j}(a_{i} , a_{j}), \forall i \in \mathcal{S}.
\label{equ_bsi}
\end{equation}
Each user \(i,j\) belongs to player \(\mathcal{S}\), and the set of strategies of player \(i\) is \(\mathcal{A_{\mathit{i}}}\). \(w_{i j}\) is a function whose variables are the elements \((a_{i},a_{j}), \,\)\(\mathcal{A_{\mathit{i}}} \times \mathcal{A_{\mathit{j}}}\). Moreover, it satisfies \(w_{i j} (a_{i}, a_{j}) = w_{j i}(a_{j}, a_{i})\) \cite{BSI_game}\cite{BSI_game_ver2}.

\subsection{POTENTIAL GAME}
A potential game is a strategy game \cite{yamamoto_potential}. In general, there may be no Nash equilibrium. Moreover, there may be several Nash equilibriums. However, we can derive a strategy profile for Nash equilibrium \cite{yamamoto_potential}. The game \(\mathscr{G}_{\mathrm{pot}} := \{\mathcal{S},\mathcal{A},(u_{i})_{i \in \mathcal{S}}\}\) is a potential game if a potential function \(\phi(\bm{a}) : \mathcal A \longmapsto \mathbb{R}\) exists such that:
\begin{equation}
\begin{split}
u_{i}(a_{i} , \bm{a_{-i}}) - u_{i}(a'_{i} , \bm{a_{-i}}) &= \phi(a_{i} , \bm{a_{-i}}) - \phi(a'_{i} , \bm{a_{-i}}), \\
& \forall a'_{i} \in \mathcal{A}, \forall i \in \mathcal{S}.
\label{equ_pot}
\end{split}
\end{equation}

\section{PROPOSED METHOD}
\subsection{ASSUMED ENVIRONMENT}

Fig. \ref{fig1} shows our assumed system environment. Here, all users are assumed to arrive in the area of \(d_{\mathrm{ver}} \times \mathit{d}_{\mathrm{hor}}\). Let \(d_{\mathrm{ver}}\) be the vertical length of the area and \(\mathit{d}_{\mathrm{hor}}\) be the horizontal length of the area. Users are managed in a centralized manner by the wireless LAN controller for AP, and when they arrive at this area, information such as the position of all users connected to the same AP and the distance between APs is exchanged \cite{centralized}. Then, on the basis of the information, the optimal user position is calculated and moved using a potential game.

In this paper, we assume that a user requesting a connection arrives at an AP. Moreover, we assume that some users are already connected to the AP. In our proposed system, the AP calculates and selects the users with the highest system throughput when they arrive at the AP. In addition, the AP prompts the selected user to move on. The user will cooperate with the AP, and this is called the user's cooperative behavior. Our -AP connection method proposed in this paper improves system throughput on the basis of the user's cooperative behavior in a multi-transmission rate environment.

\begin{figure}[t]
\centering
\includegraphics[scale=0.25]{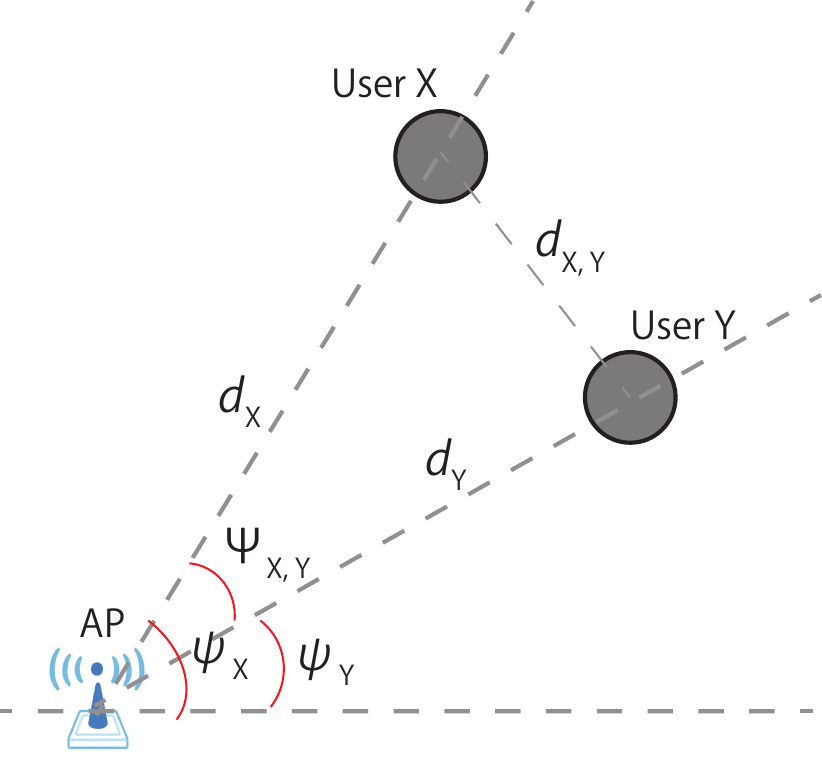}
\caption{User position model.}
\label{User_position}
\end{figure}

We assume an IEEE802.11 WLAN environment as the network environment. However, to clarify the degree of interference and the signal strength between one AP and the users connected to it, we first consider the case of a single AP in this paper. Therefore, we only consider interference received from users connected to the same AP. In other words, the interference in this paper is the effect of packet collisions between users connected to the same AP. 

Moreover, the probability of packet collisions is generally small when the number of users is small because the CSMA/CA is taken in the IEEE802.11 WLAN environment. However, when there are hidden terminals, the probability of packet collisions increases dramatically, which greatly affects the throughput obtained by users \cite{hidden_terminal_main}\cite{hidden_terminal_sub1}\cite{hidden_terminal_sub2}. This hidden terminal problem is also an important issue in real environments \cite{hidden_terminal_main}--\cite{hidden_terminal_sub2}. Therefore, we assume the case in which there are many hidden terminals in this paper.

In this network environment, each user always requests traffic in one direction upstream, and we assume a multi-rate WLAN environment where the appropriate transmission rate is selected from multiple transmission rates. The transmission rate for each user is selected appropriately according to the distance between the AP and the user and the interference from other users.

In a multi-transmission rate environment, the transmission rate is determined by the signal strength and the interference from other users. In addition, the signal strength and the interference are determined by the distance between the user and the AP. In other words, the signal strength and the interference are determined by each user's position. Therefore, the position of users is important in determining the transmission rate for each user. This position of users can be determined by the distance and angle from the AP. Fig. \ref{User_position} shows the user's position model. In this paper, let a user be \(X\). Note that the distance between a user \(X\) - AP is \(d_{X}\) and that the angle between user \(X\) - the horizontal line of AP is \(\psi_{X}\) as shown in Fig. \ref{User_position}. Moreover, the position \(\bm{d}_{X}\) of a user \(X\) can be expressed as \(\bm {d}_{X} =(d_{X}, \psi_{X})\). In particular, let the position of a moving user \(X\) be \(\bm{d}_{X}^{\mathrm{mov}}\). In this paper, we assume that users are moving. Thus, there are two types of users: users who do not move and users who move. Therefore, there are two types of user positions: before the user has moved and after the user has moved. Here, the user \(X\) position before user movement is denoted as \(\bm{\hat{d}}_{X}\), and the user \(X\) position after user movement is denoted as \(\bm{\tilde{d}}_{X} \).

\begin{figure}[t]
  \begin{minipage}[b]{0.49\linewidth}
    \centering
    \includegraphics[scale = 0.25]{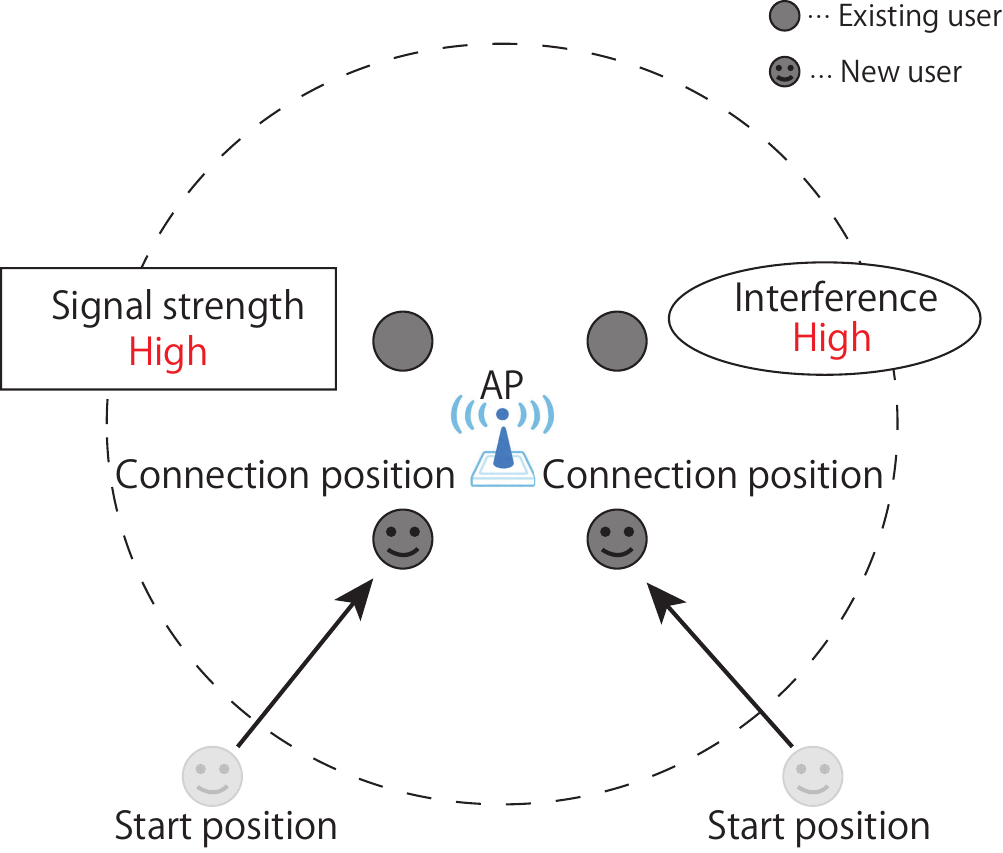}
    \subcaption{If the new user's connection position is close to the AP.}
  \end{minipage}
  \begin{minipage}[b]{0.49\linewidth}
    \centering
    \includegraphics[scale = 0.25]{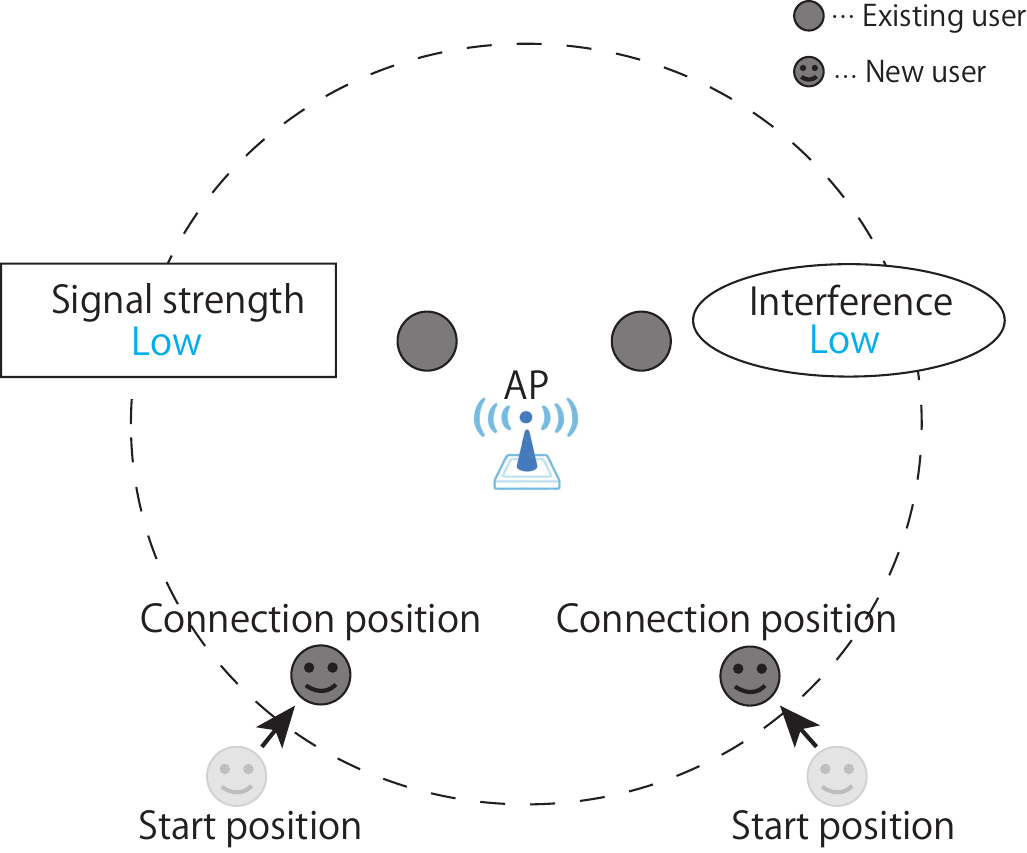}
    \subcaption{If the new user's connection position is far from the AP.}
  \end{minipage}
  \caption{Trade-off diagram between user-AP distance and interference between users.}
  \label{fig1_1}
\end{figure}

In a multi-transmission rate environment, when multiple users connect to the same AP, the transmission rate is determined by the distance \(d_{X}\) between the AP - a user \(X\) and the distance \(d_{Y}\) between the AP - the other user \(Y\). If the distance \(d_{X}\) between the AP - a user \(X\) is short, the transmission rate will be high because the signal strength will be high. However, as shown Fig. \ref{fig1_1}(a), when users move closer to the AP in order to increase the signal strength, the distance \(d_{Y}\) between the AP - the other user \(Y\) becomes shorter. As a result, the transmission rate will be low because the interference is larger when the distance \(d_{Y}\) between the AP - the other user \(Y\) is shorter. On the other hand, as shown Fig. \ref{fig1_1}(b), if the distance \(d_{X}\) between the AP - a user \(X\) is long in order to decrease the interference, the transmission rate will be low because the signal strength will be low. Thus, there is a trade-off between the distance \(d_{X}\) between an AP - a user \(X\) and the distance \(d_{Y}\) between the AP - the other user \(Y\).

In the conventional methods \cite{kato_CCNC}, there are users who were already connected to the AP (existing users) and users who newly request a connection (new users). In addition, the conventional method \cite{kato_CCNC} solves this trade-off by moving the new users to the new location. However, as shown in Fig. \ref{IN_punch}, if one existing user connects close to the AP and the other existing user connects away from the AP, there is a possibility that the high system throughput will not be achieved. This is because the conventional method \cite{kato_CCNC} assumes that only new users move. Thus, even if a new user moves to the optimal position,  as shown in Fig. \ref{IN_punch}, the transmission rate will be low because only one existing user is connected away from the AP. Therefore, in this paper, we solve this problem by extending the players in the potential game not only to new users but also to all users.

\subsection{MODELING BY GAME THEORY}
The description of each parameter in the game theory model is shown in Table \ref{Description}.
\begin{table}[t]
  \caption{Notation for our method.}
  \label{Description}
  \centering
  \scalebox{0.75}{ 
  \begin{tabular}{c||c}
    \hline \hline
    Parameters  & Description  \\
    \hline \hline
    $\mathcal{S}$ & Set of players in game theory \\ \hline
    $a_{i}$ & Strategy of player $i$ \\ \hline
    $u_{i}$ & Utility function for player $i$ \\ \hline
    $\mathcal{A}$ & Direct product set of all players' strategies \\ \hline
    $\bm{a_{-i}}$ & Set of strategies for player $j$, $j \in \mathcal{S}$, $j \neq i$ \\ \hline
    $\mathscr{G}_{\mathrm{pot}}$ & Potential game \\ \hline
    $\mathscr{G}_{\mathrm{all \, mov}}$ & All user movement game \\ \hline
    $a^{*}_{i}$ & Strategies that are Nash equilibrium for player $i$ \\ \hline
    $n$ & Moving user \\ \hline
    $\mathcal{N}$ & Set of moving users \\ \hline
    $m$ & Non-moving user \\ \hline
    $\mathcal{M}$ & Set of non-moving users \\ \hline
    $\mathcal{L}$ & Set of all users \\ \hline
    $\bm{d}_{X}$ & Position of a user $X$ \\ \hline
    $d_{X}$ & Distance between AP and a user $X$ \\ \hline
    $\psi_{X}$ & Angle between the horizontal line of AP and a user $X$ \\ \hline
    $\bm{c}^{\mathrm{mov}}$ & Combination of selected moving users \\ \hline
    $R_{X}$ & Effective transmission rate for a user $X$ \\ \hline
    $\theta$ & System throughput \\ \hline
    $W$ & Channel bandwidth \\ \hline
    $E_{X}$ & SNR (Signal to Noise Ratio) for a user $X$ \\ \hline
    $\Gamma_{X}$ & SINR (Signal to Interference and Noise Ratio) for a user $X$ \\ \hline
    $P^{\mathrm{non}}_{\mathrm{collision}}$ & Probability of non-collision \\ \hline
    $P_{\mathrm{collision}}$ & Probability of collision \\ \hline
    $P^{\mathrm{receive}}_{X}$ & Received power of an AP when a user \(X\) communicates \\ \hline
    $I_{Y}$ & Interference power from users except for a user $Y \in \mathcal{L}$, $Y \neq X$ \\ \hline
    $N_{0}$ & Background noise power \\ \hline
    $\delta$ & SINR threshold \\ \hline
    $P^{\mathrm{send}}_{\mathrm{AP}}$ & Each user's transmit power \\ \hline
    $g$ & Antenna gain \\ \hline
    $\tilde{d}_{X}^{\mathrm{mov}}$ & Distance between AP and a user $X$ after moving \\ \hline
    $\alpha$ & Path loss exponent \\ \hline
    $\bm{\tilde{d}}^{\mathrm{mov}^{*}}_{i}$ & Optimal position of a moving user $i$ \\ \hline
    $\hat{u_{i}}$ & Intermediate utility function \\
    \hline
  \end{tabular}
}
\end{table}
In this paper, the user position that maximizes system throughput is determined using a potential game, which is one of the models of game theory. A potential game has the property of having at least one solution \cite{yamamoto_potential}. This solution is called the Nash equilibrium. By using this property, the system heuristically calculates combinations of user positions. A potential game consists of three elements: the players \(\mathcal{S} = \{1,...,S\}\), the strategy \(a_{i}\) of player \(i\), and the utility function \(u_i({\bm a}):\mathcal A \longmapsto \mathbb{R}\) of player \(i\). Let \(\mathbb{R}\) be a set of real numbers. In this paper, we model our method using the potential game \(\mathscr{G}_{\mathrm {pot}}\). Moreover, the potential game \(\mathscr{G}_{\mathrm {pot}}\) is denoted as all user movement game \(\mathscr{G}_{\mathrm{all \, mov}} := \{\mathcal{S},\mathcal{A},(u_{i})_{i \in \mathcal{S}}\}\). Finally, the objective of this paper is to find the Nash equilibrium. Here, let strategies that are Nash equilibria for the player \(i\) be \(a^{*}_{i}\), and the Nash equilibrium be defined as in Eq. (\ref{equ_utility}) \cite{game_theory}.
\begin{equation}
u_{i}(a^{*}_{i} , \bm{a^{*}_{-i}}) \geq u_{i}(a'_{i} , \bm{a^{*}_{-i}}), \forall a'_{i} \in \mathcal{A}, \forall i \in \mathcal{S}.
\label{equ_utility}
\end{equation}

\begin{figure}[t]
\centering
\includegraphics[scale=0.27]{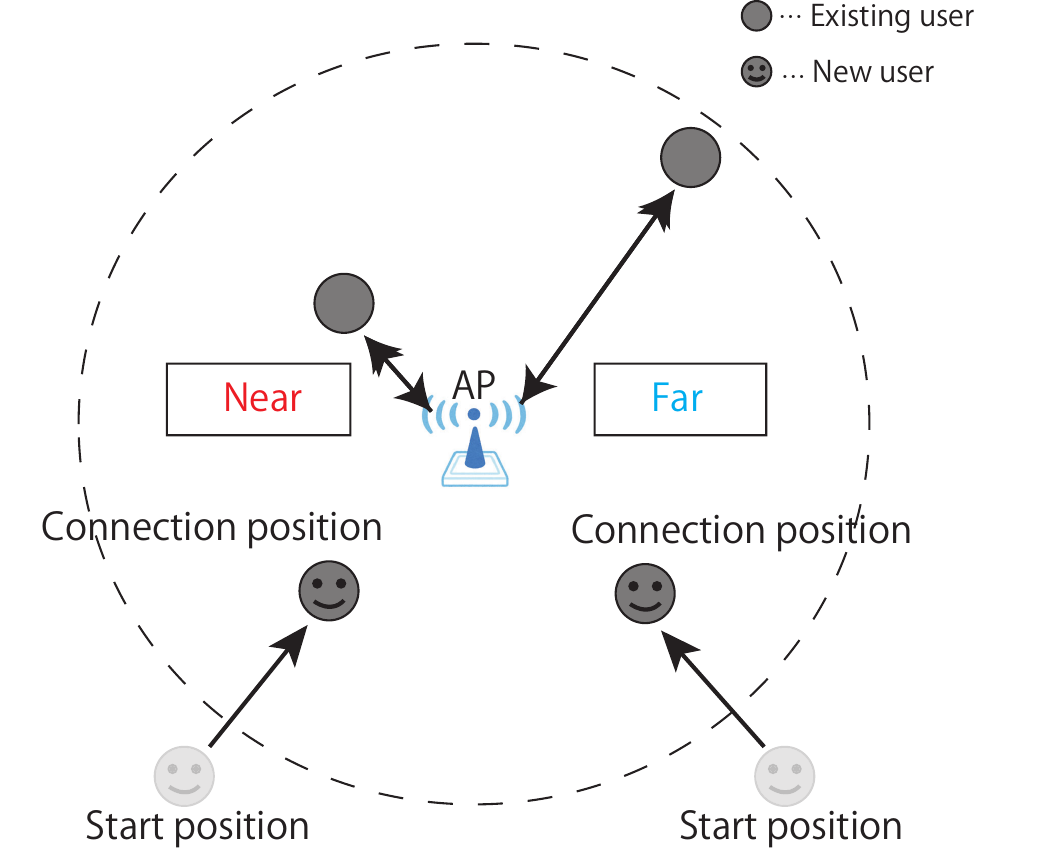}
\caption{Problem of the conventional method \cite{kato_CCNC}.}
\label{IN_punch}
\end{figure}

In modeling the potential game, we assume that there are two types of users: moving users and non-moving users. This is to simplify the problem. Therefore, the AP selects users to move and users not to move from all users connected to the AP. Let moving users be \(n \in \{1,2, \ldots ,N\} = \mathcal{N}\), non-moving users be \(m \in \{1,2, \ldots ,M\} = \mathcal{M}\), the set of all users be \(\mathcal{L} = \{1,2, \ldots ,L\}\). Moreover, we assume that combination of selected moving users are \(\bm{c}^{\mathrm{mov}} = (i,j), \, i,j \in \mathcal{L}\).

Adapting our method to the potential game, we assume that player \(\mathcal{S}\) is a set of moving users \(i \in \mathcal {N}\) and that strategy \(\mathcal{A}\) is the distance \(d_{i}\) and direction \(\psi_{i}\) of the moving users, and that the utility function \(u_{i}\) be defined as in Eq. (\ref{equ_ritoku}).
\begin{equation}
u_{i}(a_{i},\bm{a_{-i}}) =  \frac{1}{\sum_{n \in \mathcal{N}} \frac{1}{R_{n}(a_{n},\bm{a_{-n}})} + \sum_{m \in \mathcal{M}} \frac{1}{R_{m}(a_{i},\bm{a_{-i}})}},
\label{equ_ritoku}
\end{equation}
where \(R_{n}\) is the effective transmission rate obtained by a moving user and \(R_{m}\) is the effective transmission rate obtained by a non-moving user. Note that Eq. (\ref{equ_ritoku}) can consider the throughput of moving user \(i\). Thus, the system throughput \(\theta\), which means the sum of users' throughputs, is follows,
\begin{equation}
\theta = (N + M)u_{i}, i \in \mathcal {N}.
\label{equ_sys_and_thr}
\end{equation}
In particular, let system throughput when a user \(X\) and a user \(Y\) move be \(\theta_{X \, Y}\). 

From Eqs. (\ref{equ_ritoku})(\ref{equ_sys_and_thr}), maximizing \(u_{i}\) is equivalent to maximizing \(\theta\). Thus, the objective of this paper can be achieved maximizing the system throughput by maximizing the throughput of each user. Furthermore, Eq. (\ref{equ_sys_and_thr})(\ref{equ_imp}) are used as an evaluation basis. 

Here, the effective transmission rate \(R_{X}\) for a user \(X\) can be expressed from Shannon's theorem as Equ. (\ref{equ_trans}) \cite{Shannon}. However, Shannon's theorem assumes that the effect of interference from other users is constant. Therefore, because we assume a wireless LAN environment in this paper, there are cases in which packet collisions do not occur due to the CSMA/CA. In other words, if packet collisions do not occur, we do not need to consider the effects of interference. Thus, adapting Shannon's theorem to this paper, the effective transmission rate \(R_{X}\) for a user \(X\) can be expressed as follows,
\begin{equation}
\begin{split}
R_{X}&(a_{i},\bm{a_{-i}}) = P^{\mathrm{non}}_{\mathrm{collision}} \times W\log_{2}{\{1 + E_{X}(a_{i})\}}\\
& + P_{\mathrm{collision}} \times W\log_{2}{\{1 + \Gamma_{X}(a_{i},\bm{a_{-i}})\}}, i \in \mathcal{N}, X \in \mathcal{L}.
\label{equ_trans_new}
\end{split}
\end{equation}
Let \(E_{X}\) be the SNR (Signal to Noise Ratio) of a user \(X\), \(P^{\mathrm{non}}_{\mathrm{collision}}\) be probability that packets will not be collisions, \(P_{\mathrm{collision}}\) be probability of packet collisions. Moreover, \(\Gamma_{X}\) is the SINR (Signal to Interference and Noise Ratio), which is the ratio of a user \(X\)'s own signal power to the interference power; let \(\Gamma_{X}\) be the following Eq. (\ref{equ_sinr}) \cite{SINR}.
\begin{equation}
\Gamma_{X}(a_{i},\bm{a_{-i}}) = \frac{P^{\mathrm{receive}}_{X}(a_{i})}{\sum_{X \neq Y , Y \in \mathcal{L}} I_{Y}(a_{i},\bm{a_{-i}}) + N_{0}}, i \in \mathcal{N}, X \in \mathcal{L},
\label{equ_sinr}
\end{equation}
where \(P^{\mathrm{receive}}_{X}\) is the received power of an AP when a user \(X\) communicates with the AP, and \(I_{X}\) is the interference power from users except for a user \(X\). Moreover, let \(N_{0}\) be the background noise power. This SINR \(\Gamma_{X}\) is defined to exceed a threshold value \(\delta\) as shown by Eq. (\ref{SINR_equ}) in order to consider capture effects.

\(P^{\mathrm{receive}}_{X}\) is obtained by each user's transmit power \(P^{\mathrm{send}}_{\mathrm{AP}}\) \cite{9136693};
\begin{equation}
P^{\mathrm{receive}}_{X}(a_{i}) = \frac{gP^{\mathrm{send}}_{\mathrm{AP}}}{\{\tilde{d}_{X}^{\mathrm{mov}}(a_{i})\}^\alpha}, i \in \mathcal{N}, X \in \mathcal{L},
\label{equ_power}
\end{equation}
where \(g\) is the antenna gain. In this paper, users are assumed to move. This is due to the fact that the communicating target is different from the conventional method \cite{9136693}. The conventional method \cite{9136693} assumes an ad hoc network, where users communicate with each other. Thus, \(\tilde{d}_{X}^{\mathrm{mov}}\) is the distance between users. However, we assume an AP connection method. Therefore, communication is between a user and an AP. Thus, the distance \(\tilde{d}_{X}^{\mathrm{mov}}\) between a moving user after moving and an AP is used. In addition, let \(\alpha\) be the path loss exponent. This  exponent varies depending on the surrounding environment and is usually \(2 < \alpha \leq 4\) \cite{stepanov2005impact}.

Here, the effective transmission rate \(R_{X}\) of a user \(X\) can be approximated by Eq. (\ref{equ_trans_kinji}) \cite{9136693};
\begin{equation}
R_{X}(a_{i},\bm{a_{-i}}) \approx W\log_{2}{\{\Gamma_{X}(a_{i},\bm{a_{-i}})\}}, i \in \mathcal{N}, X \in \mathcal{L}.
\label{equ_trans_kinji}
\end{equation}
Moreover, the interference power \(I_{Y}\) can be expressed in the same way as in Eq. (\ref{equ_power}).

The combination of strategies that satisfy Eq. (\ref{equ_utility}) using Eq. (\ref{equ_ritoku}) calculated from (\ref{equ_trans})--(\ref{equ_power}), that is, the Nash equilibrium, derives the optimal position of the new user \(\bm{\tilde{d}}^{\mathrm{new}^{*}}_{i} = (\tilde{d}_{i}^{\mathrm{new}^{*}}, \tilde{\psi}_{i}^{\mathrm{new}^{*}})\).

Next, we prove that the all user movement game \(\mathscr{G}_{\mathrm{all \, mov}}\) is a potential game \(\mathscr{G}_{\mathrm{pot}}\) by using the BSI game. However, it is difficult to prove directly that the proposed system is a potential game using the utility function \(u_{i}\) in Eq. (\ref{equ_ritoku}) because of a difference where the sigma terms are in fractions of this equation when compared with Eq. (\ref{equ_bsi}) defined for BSI games. Thus, we redefine the utility function \(\hat{u_{i}}\), which is transformed such that the sigma term in the utility function \(u_{i}\) does not includ fractions. Define the utility function \(\hat{u_{i}}\) as:
\begin{equation}
\hat{u_{i}}(a_{i} , \bm{a_{-i}}) = \sum_{n \in \mathcal{N}} \frac{-1}{R_{n}(a_{n} , \bm{a_{-n}})} + \sum_{m \in \mathcal{M}} \frac{-1}{R_{m}(a_{j}, \bm{a_{-j}})}.
\label{equ_newritoku}
\end{equation}
Therefore, the relationship between \(u_{i}\) and \(\hat{u_{i}}\) can be expressed as Eq. (\ref{equ_relation}),
\begin{equation}
u_{i} = \frac{-1}{\hat{u_{i}}}.
\label{equ_relation}
\end{equation}
From Eq. (\ref{equ_relation}), the parameter that depends on \(u_{i}\) is \(\hat{u_{i}}\). Thus, maximizing \(\hat{u_{i}}\) is equivalent to maximizing \(u_{i}\). As a result, solving all user movement game \(\hat{\mathscr{G}}_{\mathrm {all \, mov}}\) solves all user movement game \(\mathscr{G}_{\mathrm {all \, mov}}\). Therefore, there is no problem if the utility function \(u_{i}\) changes to \(\hat{u_{i}}\).

\subsection{PROOF OF POTENTIAL GAME}
\begin{figure}[!t]
\begin{algorithm}[H]
\caption{Spatial adaptive play (SAP) algorithm. }
\label{alg1}
\begin{algorithmic}[1]	
\STATE Initially, we set \(k=0\).
\STATE All moving users collect information about the opponent's strategy.
\STATE The players are randomly selected and denoted by moving user \(i, i \in \mathcal{N}\) with the moving user's position \(\bm{\hat{d}}^{\mathrm {new}}_{i}[k]\).
\STATE Moving user \(i\) calculates the utility value \(\hat{u_{i}}(a_{i}[k+1], \bm{a_{-i}}[k])\) for all available strategies on the basis of the received strategies of the opponent.
\STATE Moving user \(i\) updates his/her position strategy according to Eq. (\ref{equ_max}), and if there is a position that satisfies Eq. (\ref{equ_max}), the moving user's position \(\bm{\hat{d}}^{\mathrm {new}}_{i} [k]\) is updated.
\STATE If there is the same maximum utility, update his/her position \(\bm{\hat{d}}^{\mathrm {new}}_{i} [k]\) to the closer one from the current position.
\STATE These steps stop if the step count \(k \) reaches a certain value; let \(\bm{\hat{d}}^{\mathrm {new}}_{i} [k]\) be \(\bm{\tilde{d}}^{\mathrm{new}^{*}}_{i}\), otherwise, they return to step 2.
\end{algorithmic}
\end{algorithm}
\end{figure}

\begin{figure}[!t]
\begin{algorithm}[H]
\caption{All user movement game algorithm.}
\label{alg2}
\begin{algorithmic}[1]
\REQUIRE{\(\bm{d}_{X}, \, X \in \mathcal{L}, \bm{d}_{Y}, \, Y \in \mathcal{L}\)}
\ENSURE{\(\bm{\tilde{d}}^{\mathrm{mov}^{*}}_{i}\)}
\STATE Initially, new users arrive in a limited square region, and we set \(\theta=0\).
\FOR{\(X=1\, \ldots \,L\)}
\FOR{\(Y=1\, \ldots \,L\)}
\IF{\(X=Y\)}
\STATE break;
\ELSIF{\(X \neq Y\)}
\STATE Users \(X,Y\) are denoted as moving users \(\bm{c}^{\mathrm{mov}} = (X,Y), \, X \neq Y, \, X,Y \in \mathcal{L}\).
\STATE Moving users \(\bm{c}^{\mathrm{mov}} = (X,Y), \, X \neq Y, \, X,Y \in \mathcal{L}\) run the SAP algorithm \cite{9136693} based on the position of user \(\bm{d}_{X}, \, X \in \mathcal{L}, \bm{d}_{Y}, \, Y \in \mathcal{L}\). They calculate the system throughput \(\theta_{X \, Y}\) and the position \(\bm{\tilde{d}}^{\mathrm{mov}}_{i} = \mathrm{argmax}\,\theta_{X \, Y}\).
\IF{\(\theta_{X \, Y} > \theta\)}
\STATE let \(\theta\) be \(\theta_{X \, Y}\) and let \(\bm{\tilde{d}}^{\mathrm{mov}^{*}}_{i}\) be \(\bm{\tilde{d}}^{\mathrm{mov}}_{i}\).
\ENDIF
\ENDIF
\ENDFOR
\ENDFOR
\end{algorithmic}
\end{algorithm}
\end{figure}

In the following, we prove that the all user movement game \(\mathscr{G}_{\mathrm{all \, mov}}\) is a potential game \(\mathscr{G}_{\mathrm{pot}}\) defined by Eq. (\ref{equ_pot}). Currently, the BSI game is sufficient for a potential game \cite{yamamoto_potential}. In addition, \(\mathscr{G}_{\mathrm{all \, mov}}\) and \(\hat{\mathscr{G}}_{\mathrm {all \, mov}}\) are equivalent. Thus, to prove that it is a potential game, we first prove that the all user movement game \(\hat{\mathscr{G}}_{\mathrm {all \, mov}}\) is a BSI game \(\mathscr{G}_{\mathrm{bsi}}\).

\begin{thm}
The all user movement game \(\mathscr{G}_{\mathrm{all \, mov}}\) is a potential game.
\end{thm}
\begin{proof}
Replacing \(u_{i}\) in Eq. (\ref{equ_bsi}) with \(\hat{u}_{i}\), the following is obtained;
\begin{equation}
\hat{u_{i}}(a_{i} , \bm{a_{-i}}) = \sum_{j \in \mathcal{N} / \{i\}} w_{i j}(a_{i} , a_{j}), \forall i \in \mathcal{N}.
\label{equ1_syomei}
\end{equation}
Let \(N\) be the number of moving users; then, the proposed system assumes \(N = 2\). In this chapter and the following, let the number of moving users \(N\) be \(N = 2\). Thus, the above equation can be rewritten as Eq. (\ref{equ2_syomei}).
\begin{equation}
\sum_{j \in \mathcal{N} / \{i\}} w_{i j}(a_{i} , a_{j}) = w_{i j}(a_{i} , a_{j}), j \in \mathcal{N} \backslash {\{i\}}.
\label{equ2_syomei}
\end{equation}
To make \(w_{i j}\) similar to the utility function \(\hat{u_{i}}\), replace \(w_{i j}\) with the following.
\begin{equation}
w_{i j}(a_{i} , a_{j}) = \sum_{k \in \mathcal{N}} \frac{-1}{R_{k}(a_{k} , \bm{a_{-k}})} + \sum_{l \in \mathcal{M}} \frac{-1}{R_{l}(a_{i}. \bm{a_{-i}})}, j \in \mathcal{N} \backslash {\{i\}}
\label{equ3_syomei}
\end{equation}
The utility function \(\hat{u_{i}}\) of Eq. (\ref{equ1_syomei}) is as:
\begin{equation}
\begin{split}
\hat{u_{i}}(a_{i} , \bm{a_{-i}}) &= \sum_{n \in \mathcal{N}} \frac{-1}{R_{n}(a_{n} , \bm{a_{-n}})}\\ 
&+ \sum_{m \in \mathcal{M}} \frac{-1}{R_{m}(a_{i} , \bm{a_{-i}})}, j \in \mathcal{N} \backslash {\{i\}}.
\label{equ4_syomei}
\end{split}
\end{equation}
To prove the BSI game, it is necessary to check that the function \(w_{i j}\) satisfies \(w_{i j} (a_{i}, a_{j}) = w_{j i}(a_{j}, a_{i})\).
\begin{equation}
\begin{split}
w_{i j} (a_{i}, a_{j}) &= \sum_{n \in \mathcal{N}} \frac{-1}{R_{n}(a_{n} , \bm{a_{-n}})}\\ 
&+ \sum_{m \in \mathcal{M}} \frac{-1}{R_{m}(a_{i} , \bm{a_{-i}})}, j \in \mathcal{N} \backslash {\{i\}}, i \in \mathcal{N} \backslash {\{j\}},
\label{equ5_syomei}
\end{split}
\end{equation}
\begin{equation}
\begin{split}
w_{j i} (a_{j}, a_{i}) &= \sum_{n \in \mathcal{N}} \frac{-1}{R_{n}(a_{n} , \bm{a_{-n}})}\\ 
&+ \sum_{m \in \mathcal{M}} \frac{-1}{R_{m}(a_{j} , \bm{a_{-j}})}, j \in \mathcal{N} \backslash {\{i\}}, i \in \mathcal{N} \backslash {\{j\}},
\label{equ6_syomei}
\end{split}
\end{equation}
where functions \(w_{i j}\) and \(w_{j i}\) are Eqs. (\ref{equ5_syomei}) and (\ref{equ6_syomei}), respectively, since the second term on the right hand side is different from Eqs. (\ref{equ5_syomei})(\ref{equ6_syomei}). If \(R_{m} (a_{i}, \bm{a_{-i}}) = R_{m} (a_{j} , \bm{a_{-j}})\) holds, then the function \(w_{i j}\) also holds. It can be verified that \(w_{i j} (a_{i}, a_{j}) = w_{j i}(a_{j}, a_{i})\) is satisfied of condition of BSI game. In this system model, \(i\) and \(j, j \in \mathcal{N} \backslash {\{i\}}, i \in \mathcal{N} \backslash {\{j\}}\), which correspond to moving users, are only two moving users, and they interact with each other because they interfere with each other \cite{9136693}. Thus, \(R_{m} (a_{i}, a_{j}) = R_{m} (a_{j}, a_{i}), \forall m \in \mathcal{M}\). As a result, from Eqs. (\ref{equ4_syomei})(\ref{equ5_syomei})(\ref{equ6_syomei}), we can say that the all user movement game \(\hat{\mathscr{G}}_{\mathrm {all \, mov}}\) is a BSI game \(\mathscr{G}_{\mathrm{bsi}}\). Since the BSI game is a sufficient condition for a potential game, the all user movement game \(\hat{\mathscr{G}}_{\mathrm {all \, mov}}\) is a potential game \(\mathscr{G}_{\mathrm{pot}}\). As a result, \(\mathscr{G}_{\mathrm{all \, mov}}\) is a potential game \(\mathscr{G}_{\mathrm{pot}}\).
\end{proof}

\subsection{PROPOSED ALGORITHM}
\begin{table}[!t]
  \caption{Parameter values.}
  \label{Para}
  \centering
  \begin{tabular}{cc}
    \hline
    Parameter  & value  \\
    \hline
    Vertical length of area $d_{\mathrm{ver}}$  & $60 \, \mathrm{m}$ \\
    Horizontal length of area $d_{\mathrm{hor}}$  & $60 \, \mathrm{m}$ \\
    Communication standard & IEEE 802.11g \\
    Number of APs & 1 \\
    Coordinates of AP $\bm{d}_{\mathrm{AP}}$ & $(30,30)\, \mathrm{m}$ \\
    Number of moving users $N$ & 2 \\
    Number of non-moving users $M$ & 2 \\
    Path loss exponent $\alpha$ & 2 \\
    Transmit power of terminal $P^{\mathrm{send}}$ & $32 \, \mathrm{dBm}$ \\
    Antenna gain $g$ & 5.00 \\
    SINR thresholds $\delta$ & $-20 \, \mathrm{dB}$ \\
    White Gaussian noise $N_{0}$ & $10^{-13} \, \mathrm{W}$ \\
    Channel bandwidth $W$ & $20 \, \mathrm{MHz}$ \\
    Probability of non-collision $P^{\mathrm{non}}_{\mathrm{collision}}$ & $3 \, \%$ \\
    Probability of collision $P_{\mathrm{collision}}$ & $97 \, \%$ \\
    Step count $k$ & 1000 \\
    \hline
  \end{tabular}
\end{table}

\begin{figure*}[!t]
      \begin{minipage}[t]{0.33\linewidth}
        \centering
        \includegraphics[keepaspectratio, scale=0.2]{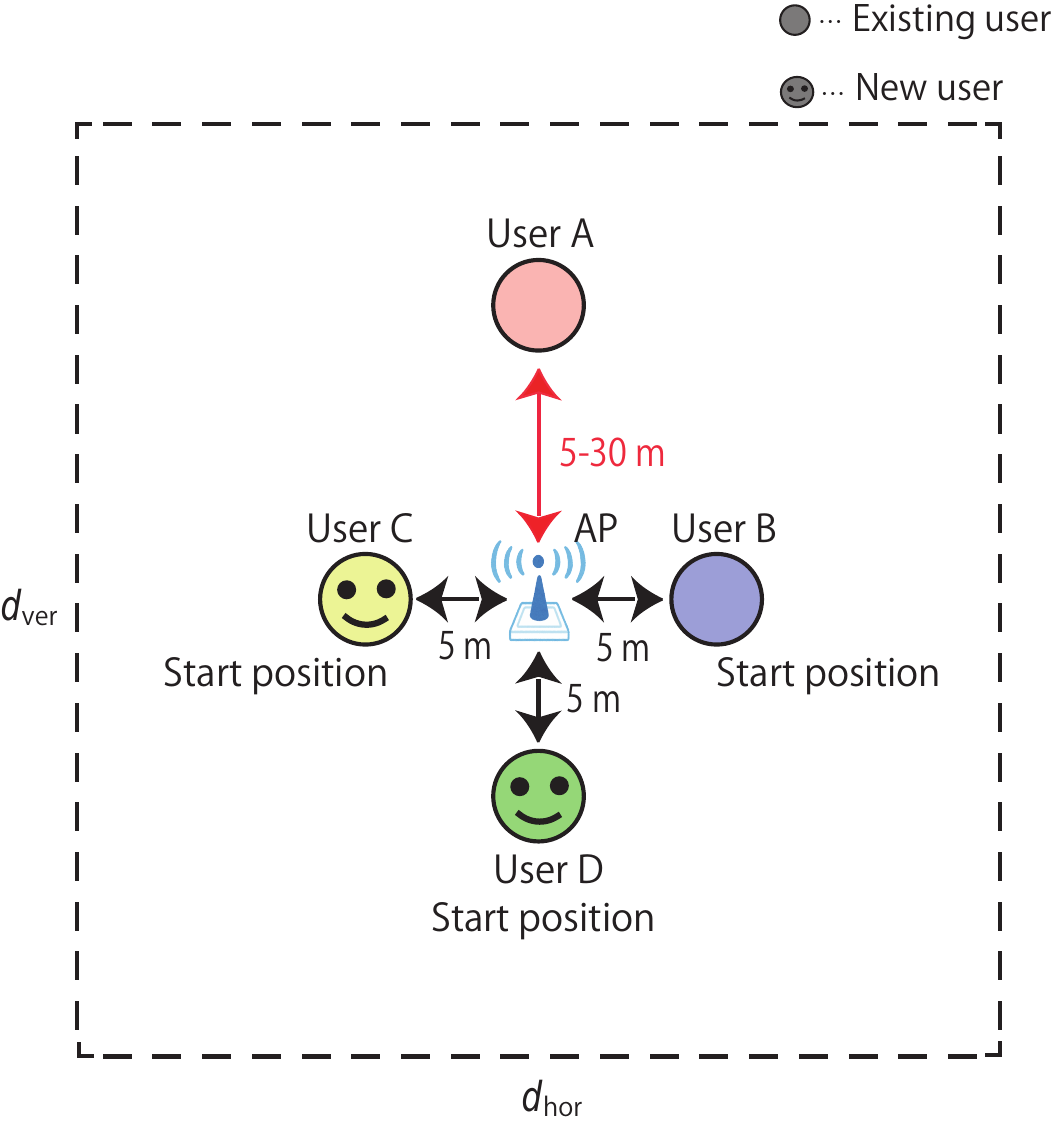}
        \subcaption{Assumed environment of pattern \((\mathrm{I})\).}
      \end{minipage}
      \begin{minipage}[t]{0.33\linewidth}
        \centering
        \includegraphics[keepaspectratio, scale=0.2]{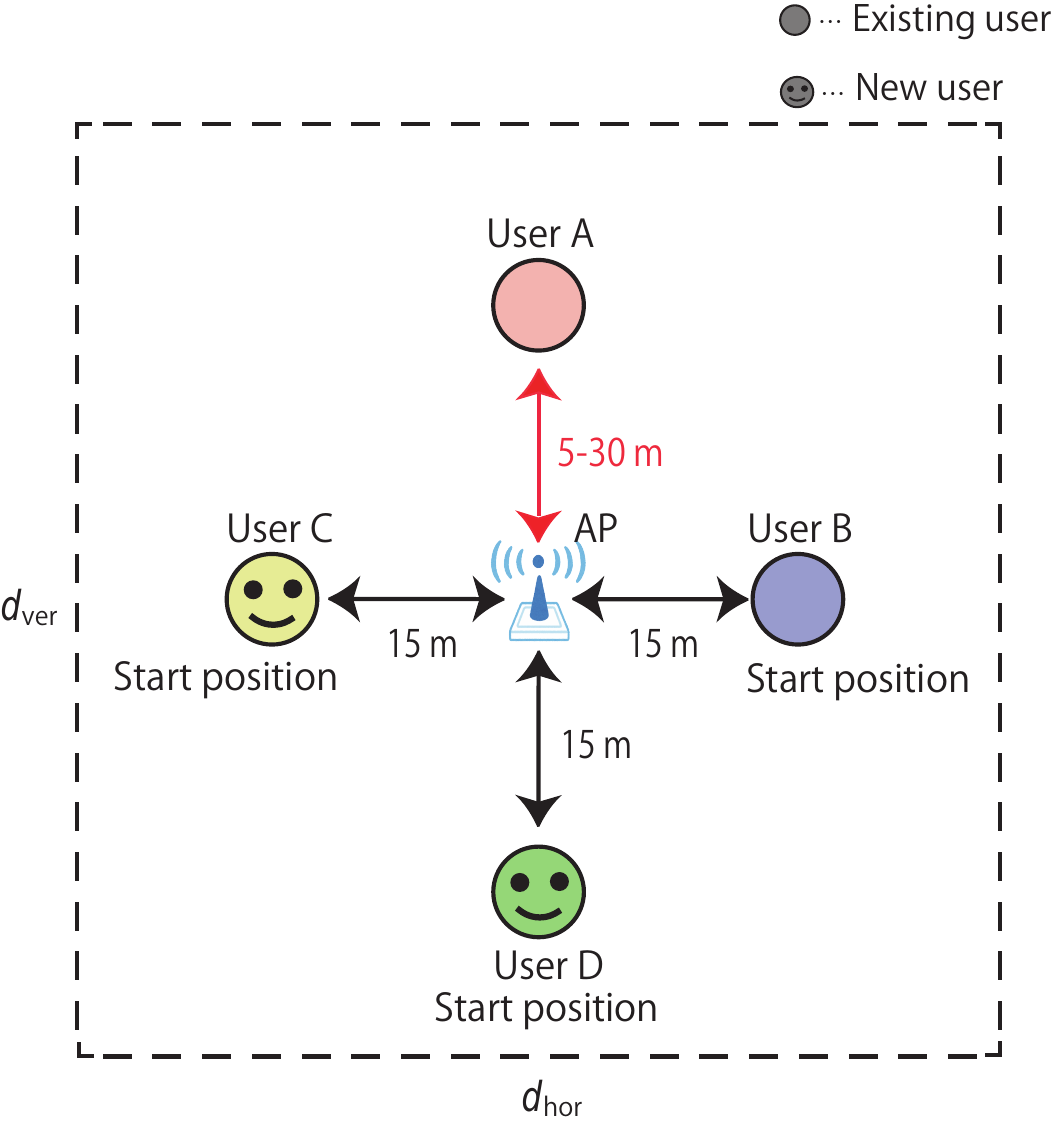}
        \subcaption{Assumed environment of pattern \((\mathrm{I}\hspace{-1.2pt}\mathrm{I})\).}
      \end{minipage}  
      \begin{minipage}[t]{0.33\linewidth}
        \centering
        \includegraphics[keepaspectratio, scale=0.2]{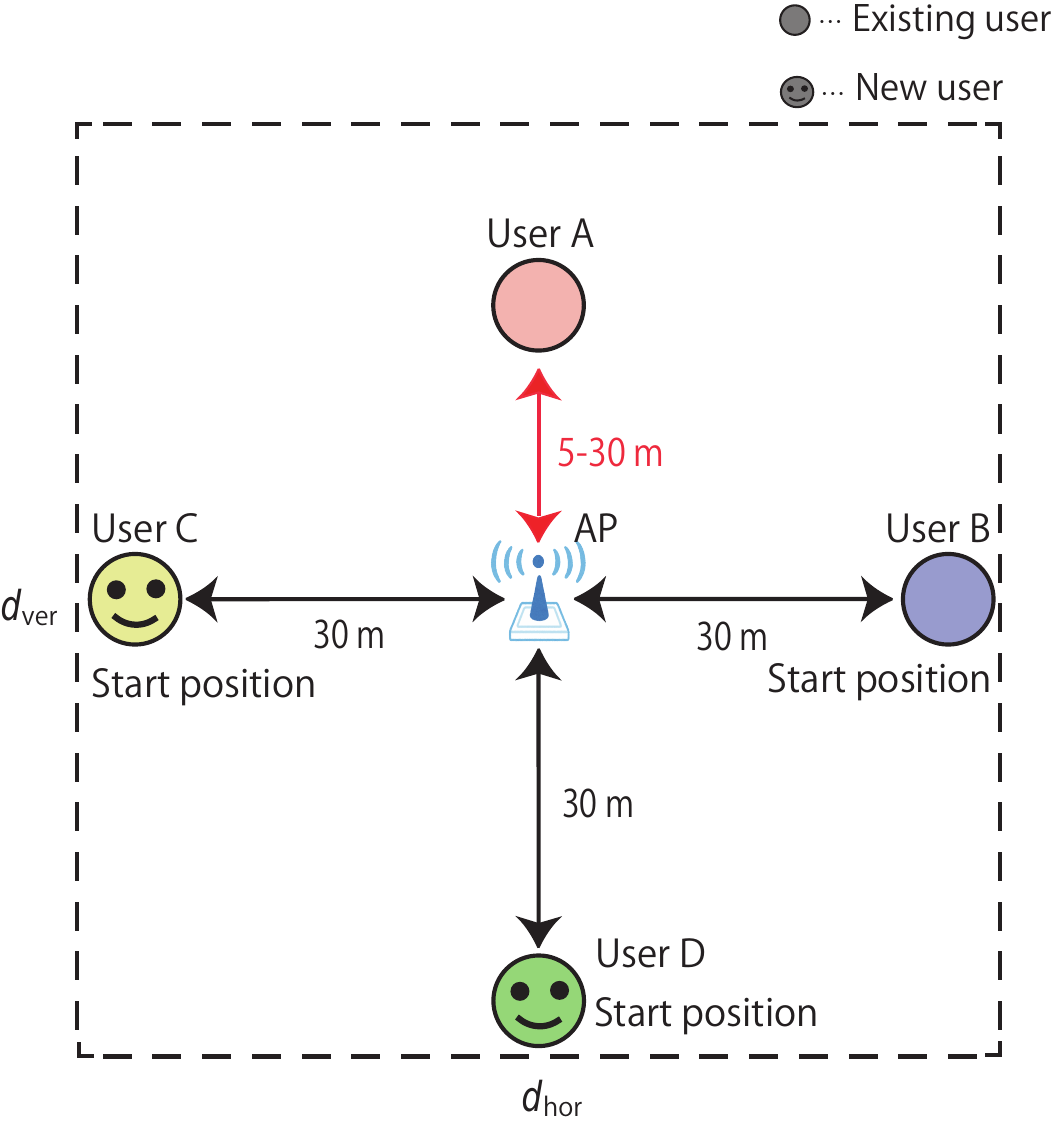}
        \subcaption{Assumed environment of pattern \((\mathrm{I}\hspace{-1.2pt}\mathrm{I}\hspace{-1.2pt}\mathrm{I})\).}
      \end{minipage}
      \begin{minipage}[t]{0.33\linewidth}
        \centering
        \includegraphics[keepaspectratio, scale=0.2]{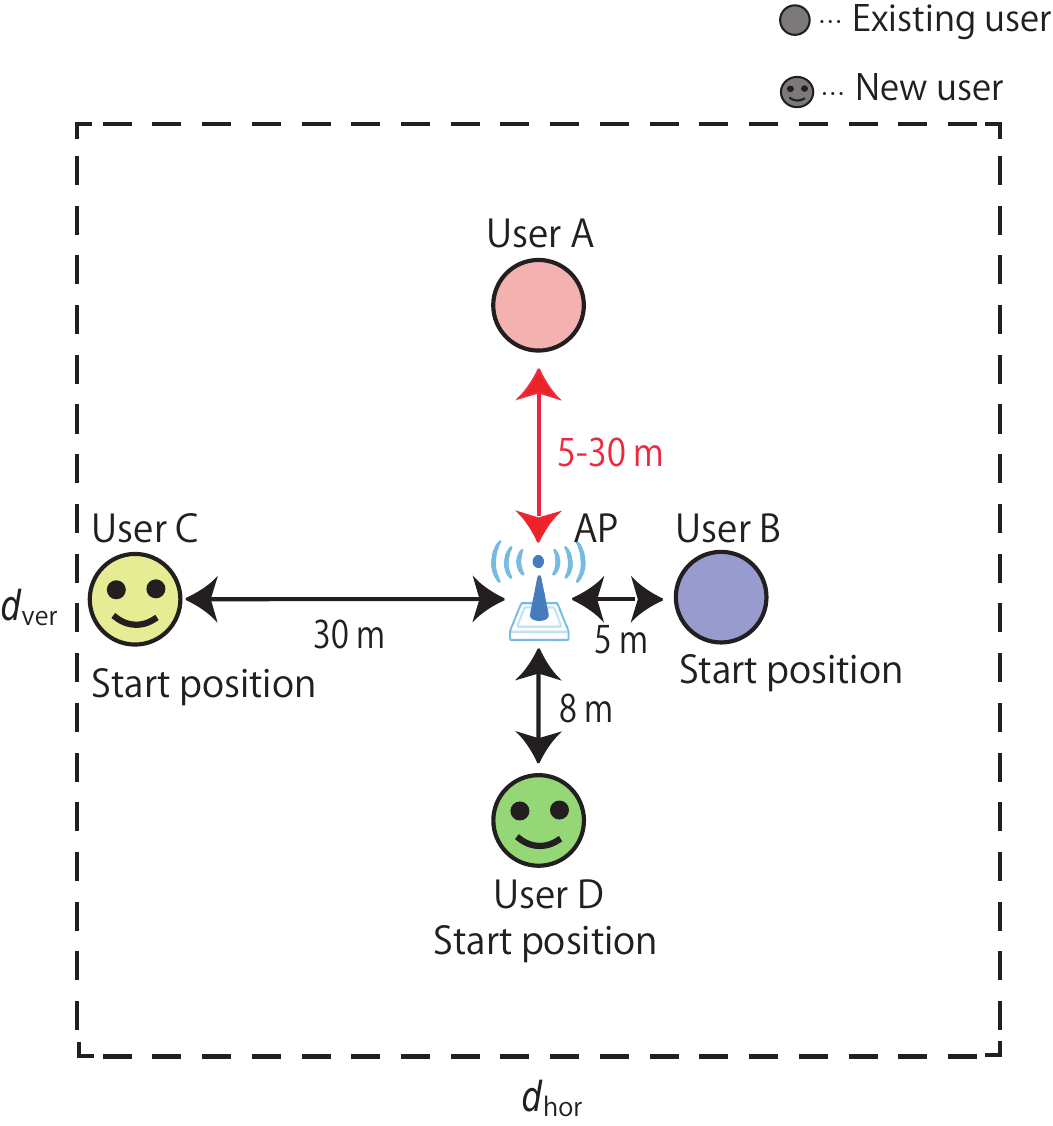}
        \subcaption{Assumed environment of pattern \((\mathrm{I}\hspace{-1.2pt}\mathrm{V})\).}
      \end{minipage}
      \begin{minipage}[t]{0.33\linewidth}
        \centering
        \includegraphics[keepaspectratio, scale=0.2]{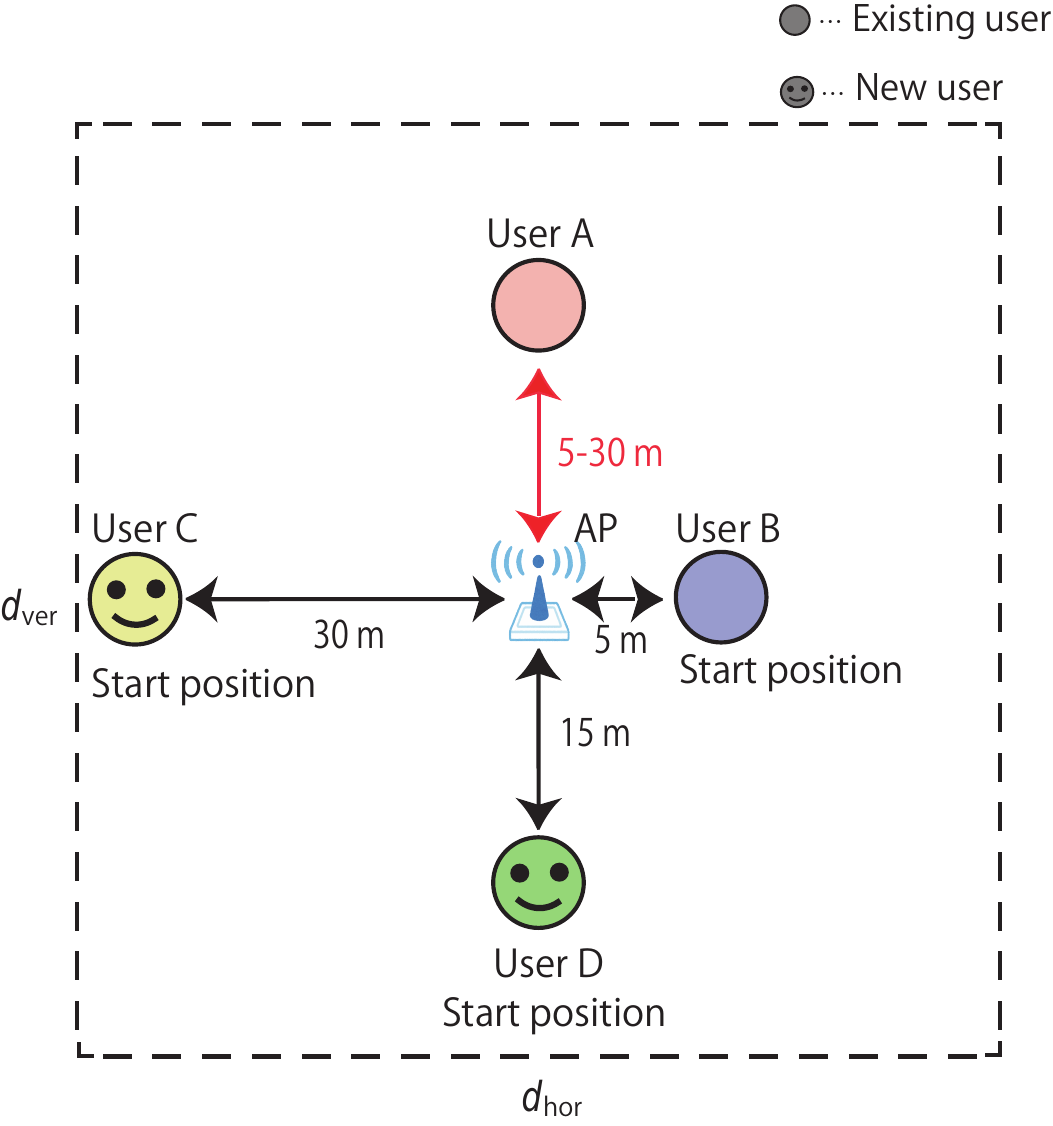}
        \subcaption{Assumed environment of pattern \((\mathrm{V})\).}
      \end{minipage}
      \begin{minipage}[t]{0.33\linewidth}
        \centering
        \includegraphics[keepaspectratio, scale=0.2]{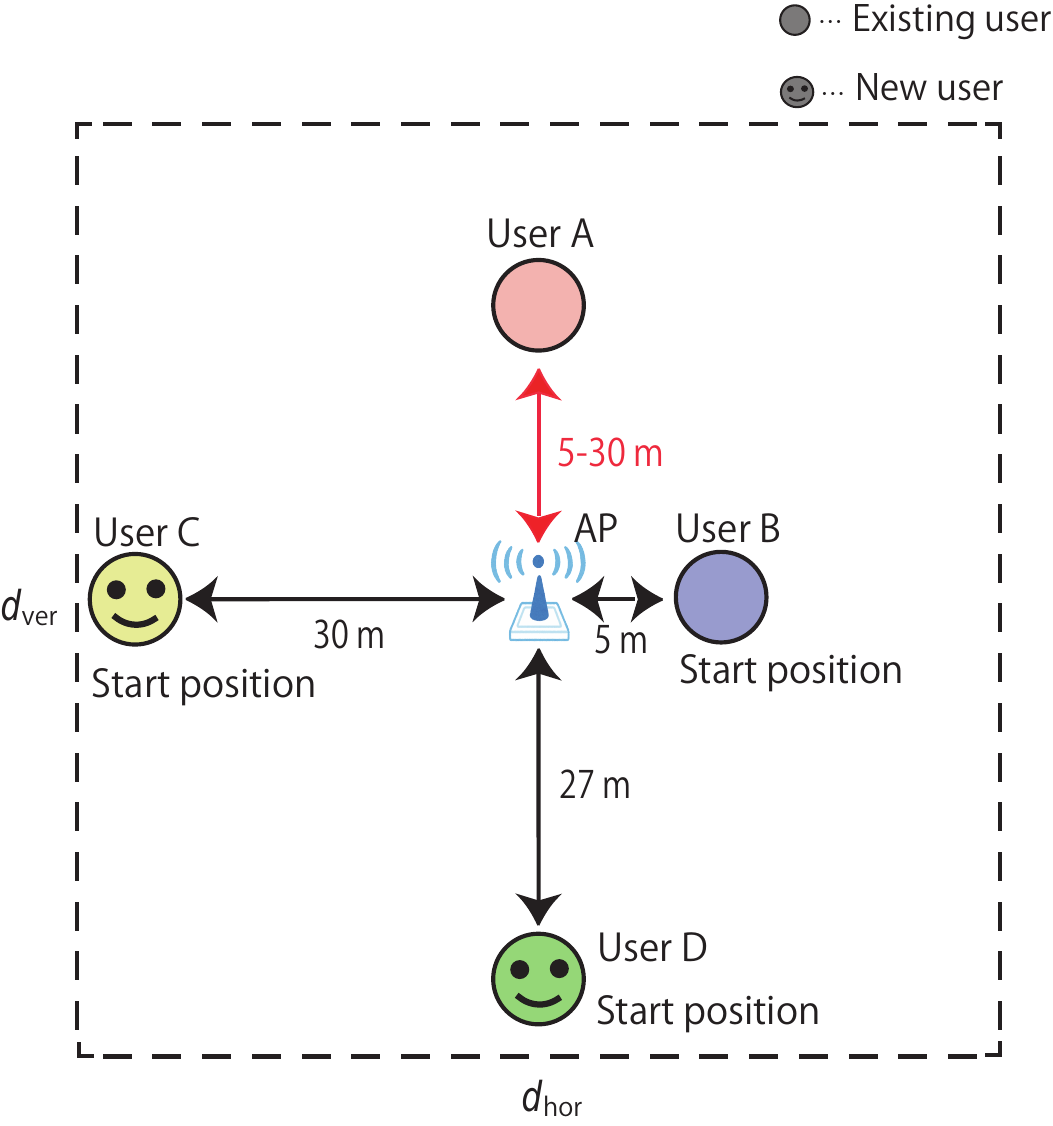}
        \subcaption{Assumed environment of pattern \((\mathrm{V}\hspace{-1.2pt}\mathrm{I})\).}
      \end{minipage}
     \caption{Assumed environment.}
     \label{Assumed_enviroment}
\end{figure*}

There are multiple algorithms for finding Nash equilibrium in potential games \cite{yamamoto_potential}. However, spatial adaptive play (SAP) maximizes the potential function with a high probability of the Nash equilibrium converging. In this paper, we use an SAP algorithm \cite{9136693}. This method decides a player's next strategy probabilistically.

Player \(i \in \mathcal N\) can only change his/her strategy \(a_{i}[k]\) at the \(k\)th iteration time. Player \(i\) chooses at the \(k+1\)th iteration time strategy \(a_{i}[k+1]\) with probability \(p_{i}\) as follow \cite{yamamoto_potential}.
\begin{equation}
p_{i}(a_{i}[k+1], \bm{a_{-i}}[k]) = \frac{\exp\{ \beta \hat{u_{i}}(a_{i}[k+1], \bm{a_{-i}}[k]) \}}{\sum_{a_{i}^{'} \in \mathcal{A}_{i} \backslash \{a_{i}\} } \exp \{ \beta \hat{u_{i}}(a_{i}^{'}, \bm{a_{-i}} [k]) \}},
\label{equ_pro}
\end{equation}
where \(\beta\) is a parameter that takes the value \(0<\beta<\infty\); in this calculation, we assume that \(\beta = k\)  \cite{9136693}. Using this \(p_{i}\), player \(i\) chooses the strategy that can maximize his/her utility value \(\hat{u_{i}}\).

The maximum utility value can be written using \(p_{i}\) as in \cite{9136693}:
\footnotesize
\begin{equation}
\begin{split}
\max \hat{u_{i}}(a_{i}[k+1], \bm{a_{-i}}[k]) \approx \max [&p_{i}(a_{i}[k+1])\hat{u_{i}}(a_{i}[k+1], \bm{a_{-i}}[k])\\ 
&- \frac{1}{\beta}p_{i}(a_{i}[k+1])\log p_{i}(a_{i}[k+1])].
\label{equ_max}
\end{split}
\end{equation}
\normalsize
Eq. (\ref{equ_max}) is used to find the strategy that maximizes the utility value. Details on the SAP algorithm and the all user movement game algorithm are given in Algorithm \ref{alg1} and Algorithm \ref{alg2}. The Nash equilibrium of our proposed game \(\mathscr{G}_{\mathrm{all \, mov}}\) is derived by these algorithms.

\section{NUMERICAL ANALYSIS}

\begin{figure*}[!t]
  \begin{minipage}[b]{0.33\linewidth}
    \centering
    \includegraphics[keepaspectratio, scale=0.35]{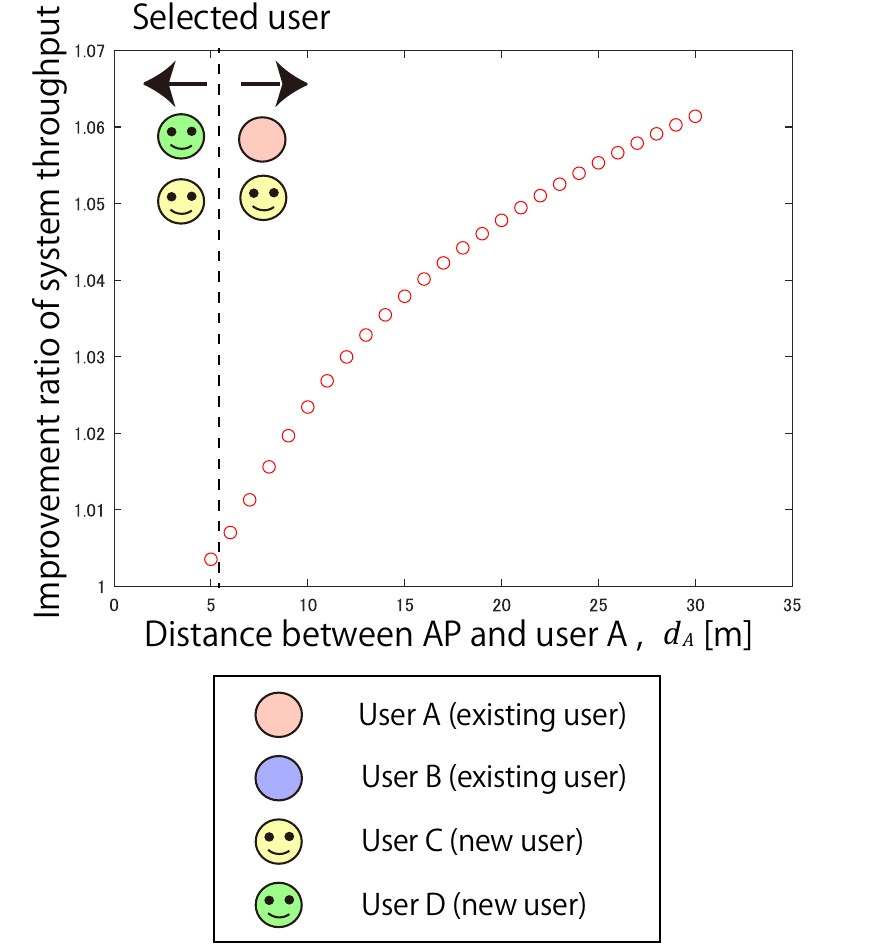}
    \subcaption{Improvement ratio of system throughput between the proposed method and the case in which users do not move.}
  \end{minipage}
  \begin{minipage}[b]{0.33\linewidth}
    \centering
    \includegraphics[keepaspectratio, scale=0.35]{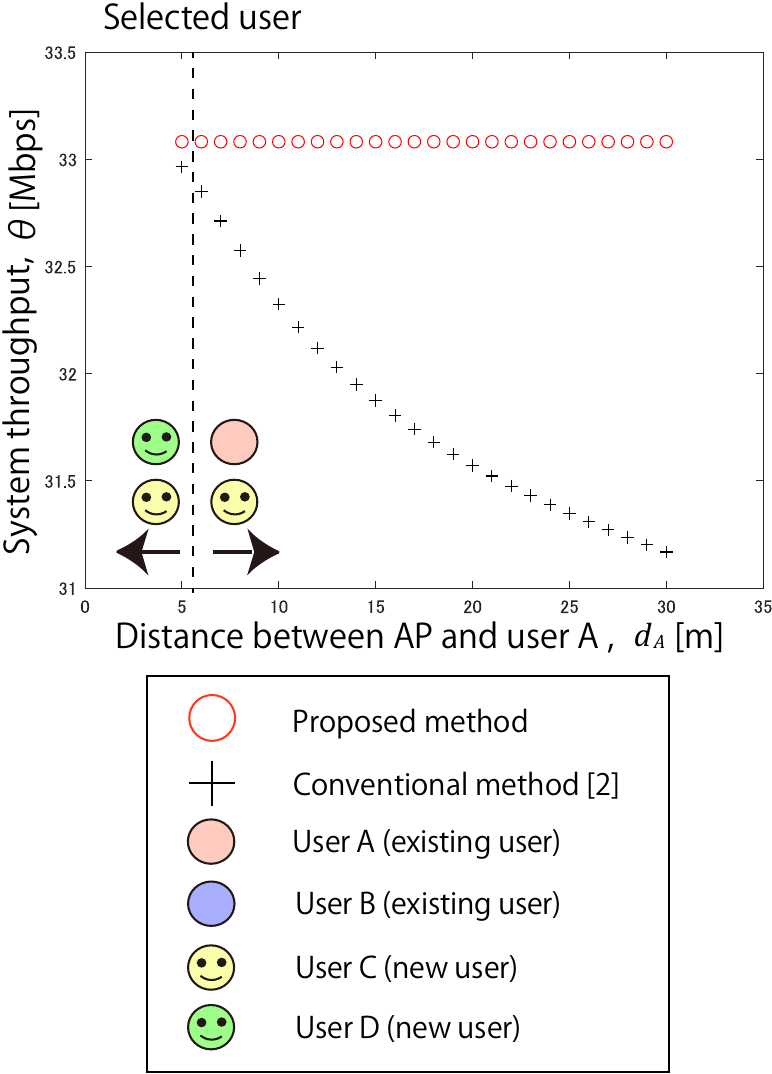}
    \subcaption{System throughput of proposed and conventional methods \cite{miyata_CCNC}.}
  \end{minipage}
  \begin{minipage}[b]{0.33\linewidth}
    \centering
    \includegraphics[keepaspectratio, scale=0.35]{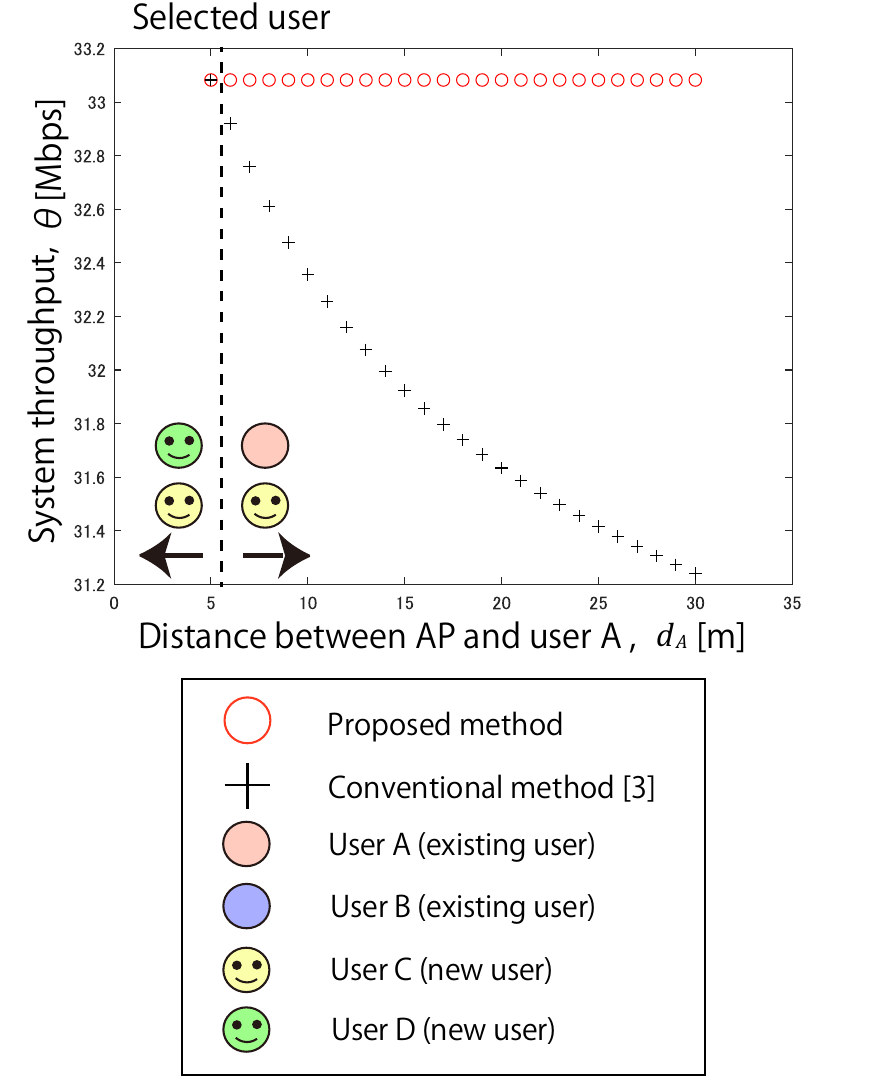}
    \subcaption{System throughput of proposed and conventional methods \cite{kato_CCNC}.}
  \end{minipage}
  \caption{System throughput of each method in pattern \((\mathrm{I})\).}
  \label{System_throughput_comparison_ver1}
\end{figure*}

\begin{figure*}[!t]
  \begin{minipage}[b]{0.33\linewidth}
    \centering
    \includegraphics[keepaspectratio, scale=0.35]{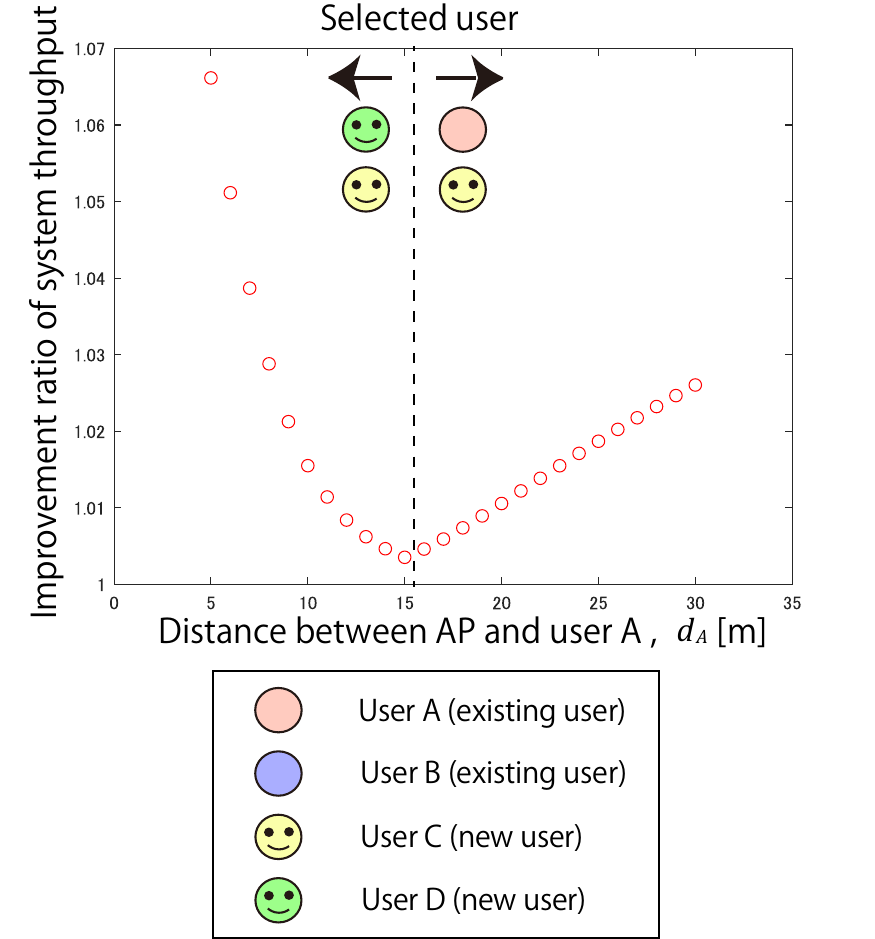}
    \subcaption{Improvement ratio of system throughput between the proposed method and the case in which users do not move.}
  \end{minipage}
  \begin{minipage}[b]{0.33\linewidth}
    \centering
    \includegraphics[keepaspectratio, scale=0.35]{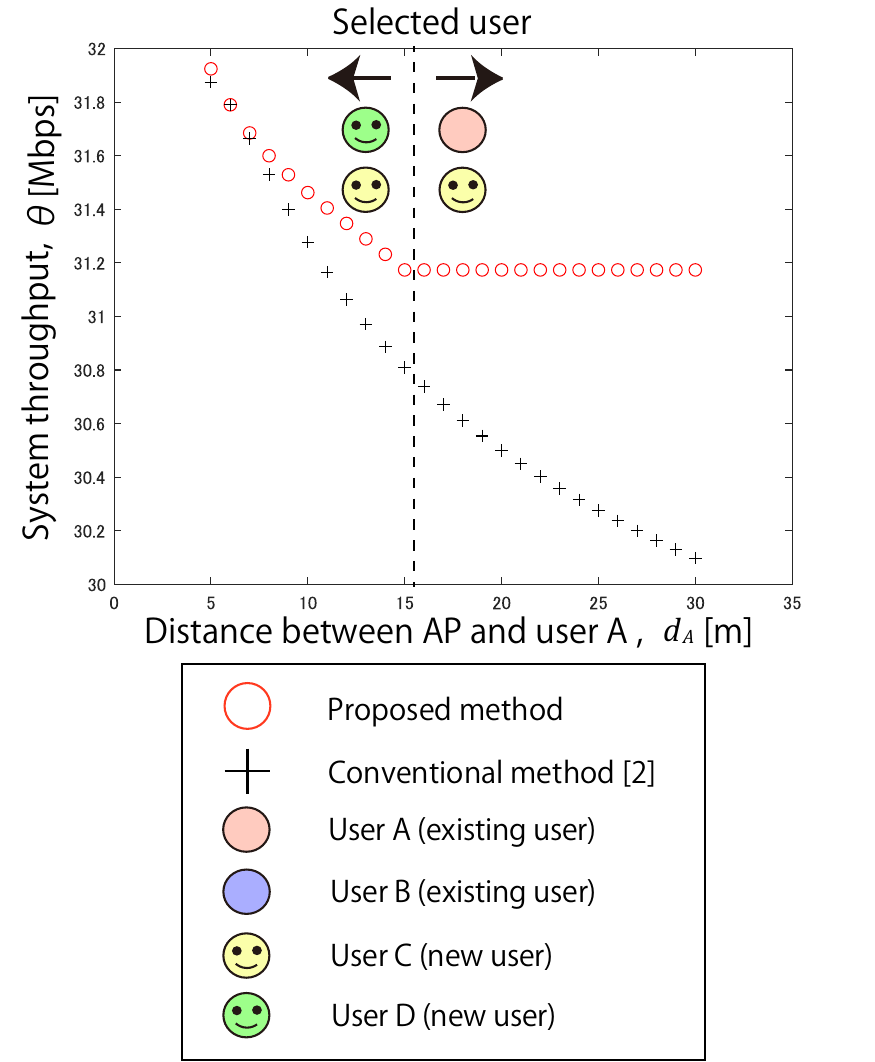}
    \subcaption{System throughput of proposed and conventional methods \cite{miyata_CCNC}.}
  \end{minipage}
  \begin{minipage}[b]{0.33\linewidth}
    \centering
    \includegraphics[keepaspectratio, scale=0.35]{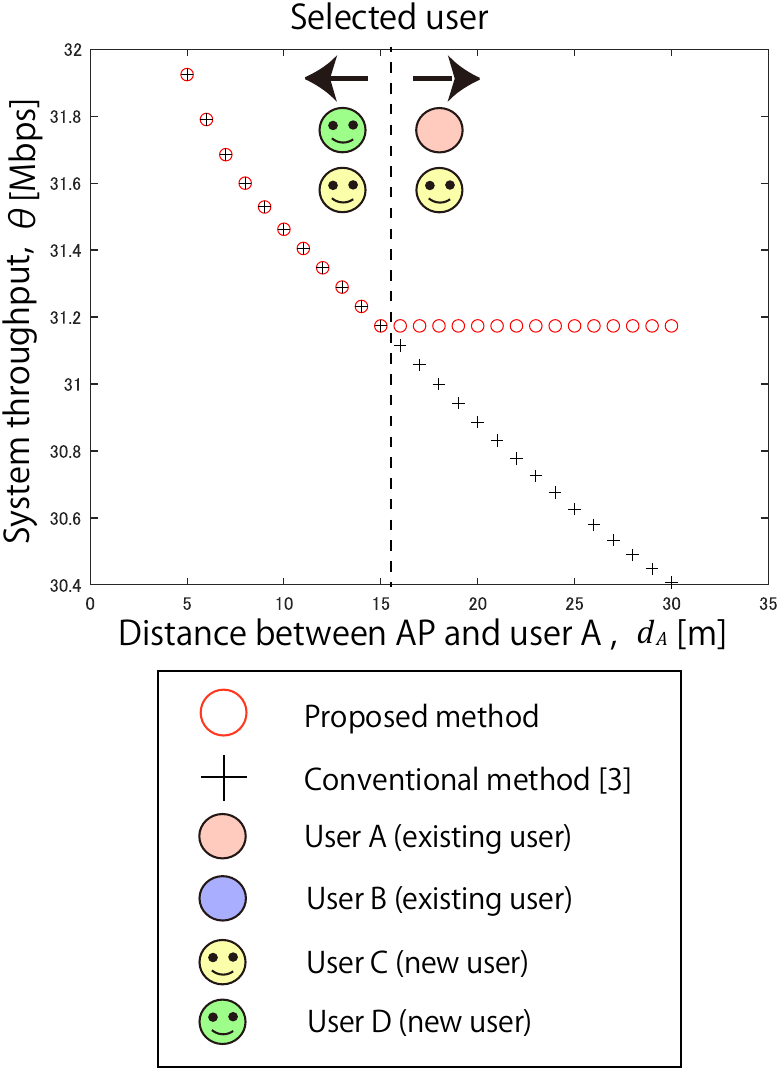}
    \subcaption{System throughput of proposed and conventional methods \cite{kato_CCNC}.}
  \end{minipage}
  \caption{System throughput of each method in pattern \((\mathrm{I}\hspace{-1.2pt}\mathrm{I})\).}
  \label{System_throughput_comparison_ver2}
\end{figure*}

\begin{figure*}[!t]
  \begin{minipage}[b]{0.33\linewidth}
    \centering
    \includegraphics[keepaspectratio, scale=0.35]{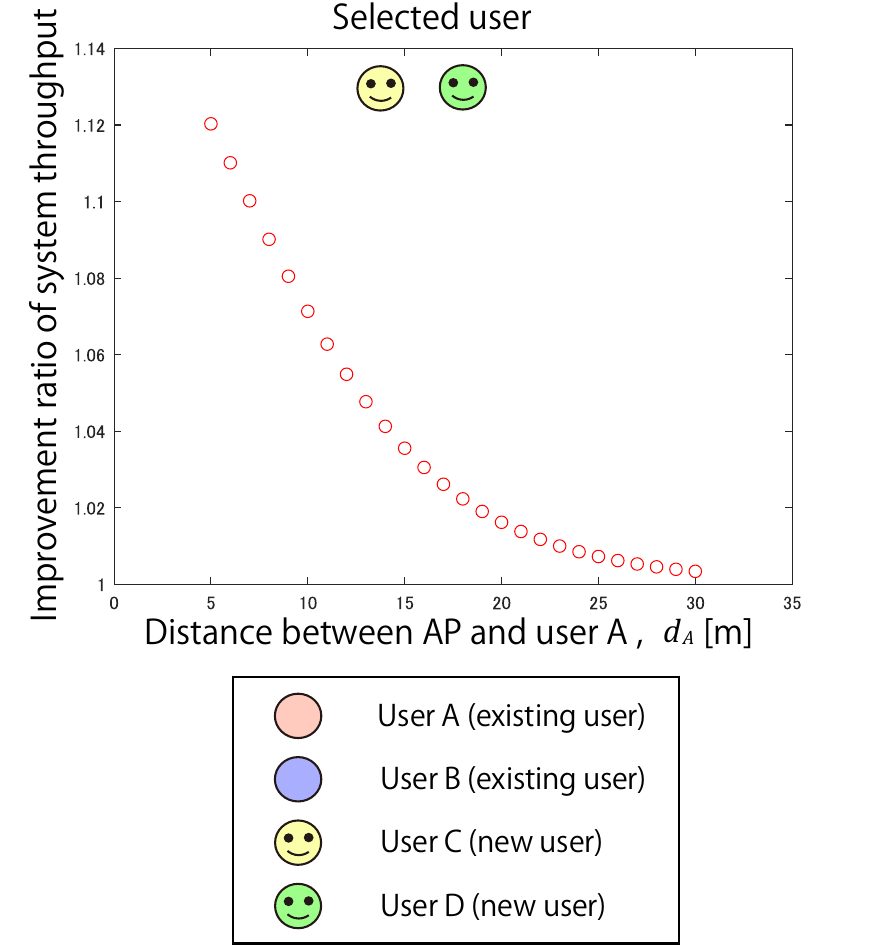}
    \subcaption{Improvement ratio of system throughput between the proposed method and the case in which users do not move.}
  \end{minipage}
  \begin{minipage}[b]{0.33\linewidth}
    \centering
    \includegraphics[keepaspectratio, scale=0.35]{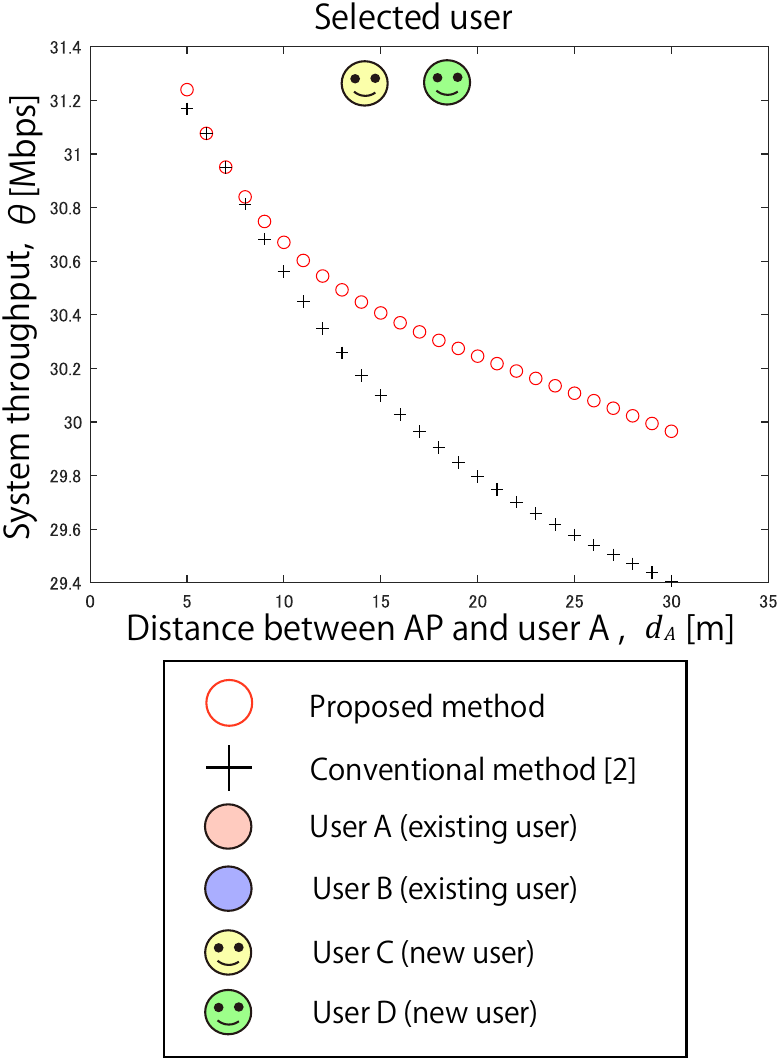}
    \subcaption{System throughput of proposed and conventional methods \cite{miyata_CCNC}.}
  \end{minipage}
  \begin{minipage}[b]{0.33\linewidth}
    \centering
    \includegraphics[keepaspectratio, scale=0.35]{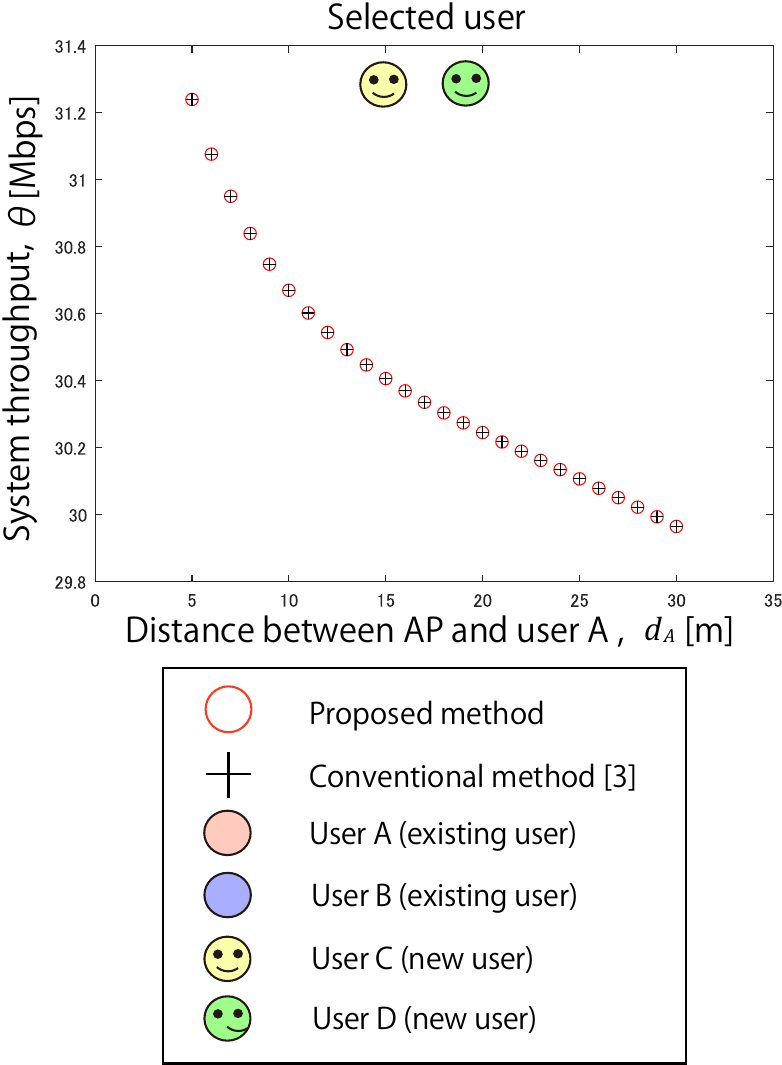}
    \subcaption{System throughput of proposed and conventional methods \cite{kato_CCNC}.}
  \end{minipage}
  \caption{System throughput of each method in pattern \((\mathrm{I}\hspace{-1.2pt}\mathrm{I}\hspace{-1.2pt}\mathrm{I})\).}
  \label{System_throughput_comparison_ver3}
\end{figure*}

\begin{figure}[!t]
  \begin{minipage}[b]{0.48\linewidth}
    \centering
    \includegraphics[keepaspectratio, scale=0.3]{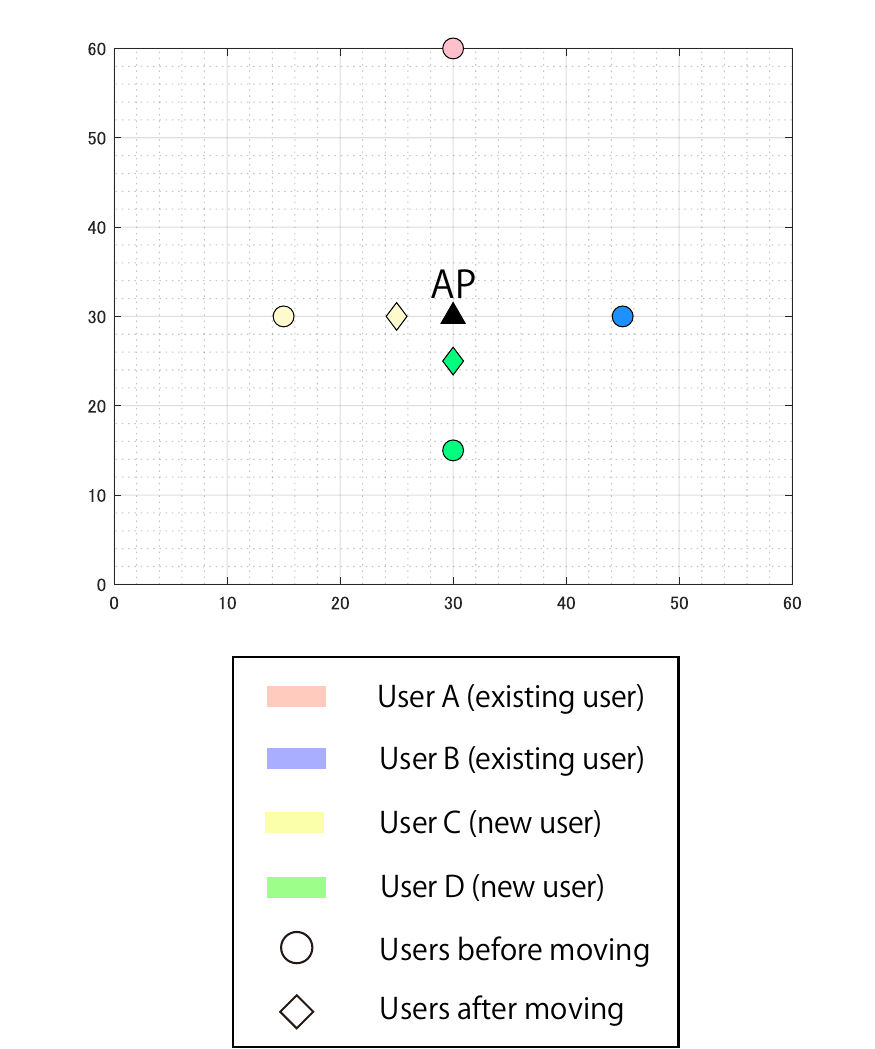}
    \subcaption{Optimal position of each user in conventional method \cite{miyata_CCNC}.}
  \end{minipage}
  \begin{minipage}[b]{0.48\linewidth}
    \centering
    \includegraphics[keepaspectratio, scale=0.3]{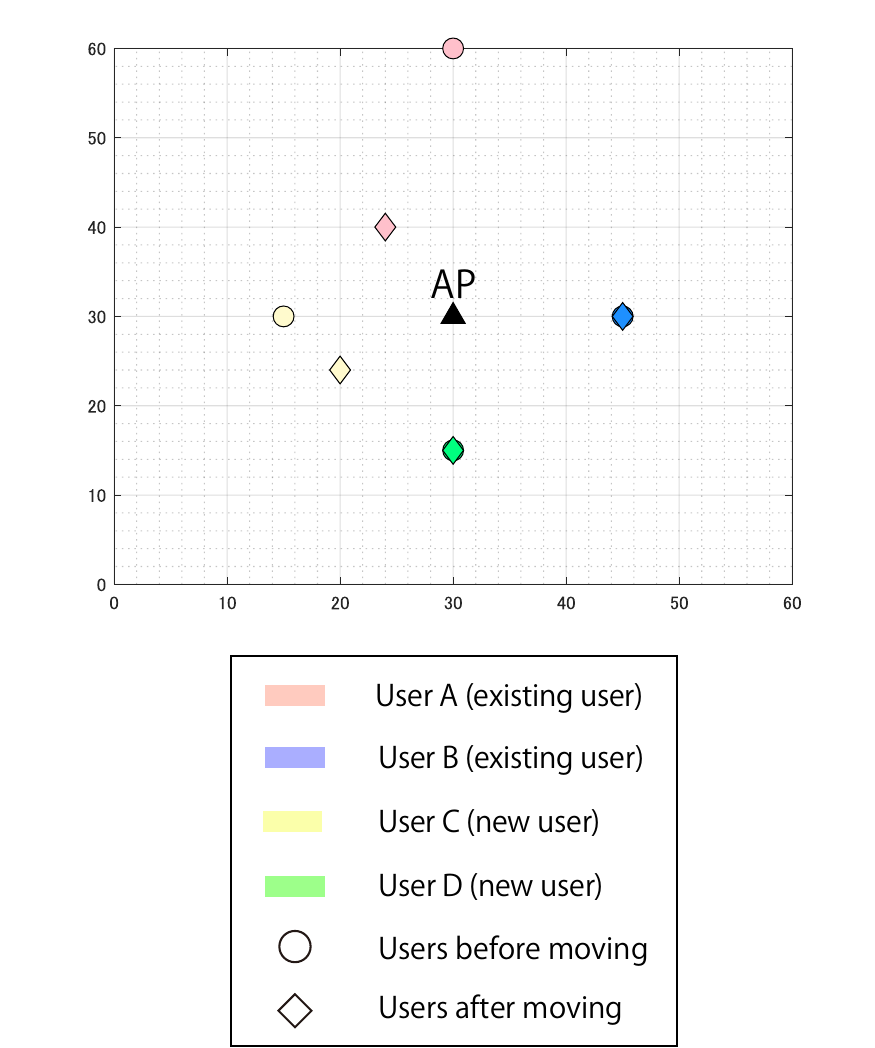}
    \subcaption{Optimal position of each user in proposed method.}
  \end{minipage}
  \caption{Optimal position of each user at \(\hat{\bm{d}}_{A} = (30, 90^\circ)\) in pattern \((\mathrm{I}\hspace{-1.2pt}\mathrm{I})\).}
  \label{User_saiteki_position_pattern2_60}
\end{figure}

\begin{figure}[!t]
  \begin{minipage}[b]{0.48\linewidth}
    \centering
    \includegraphics[keepaspectratio, scale=0.3]{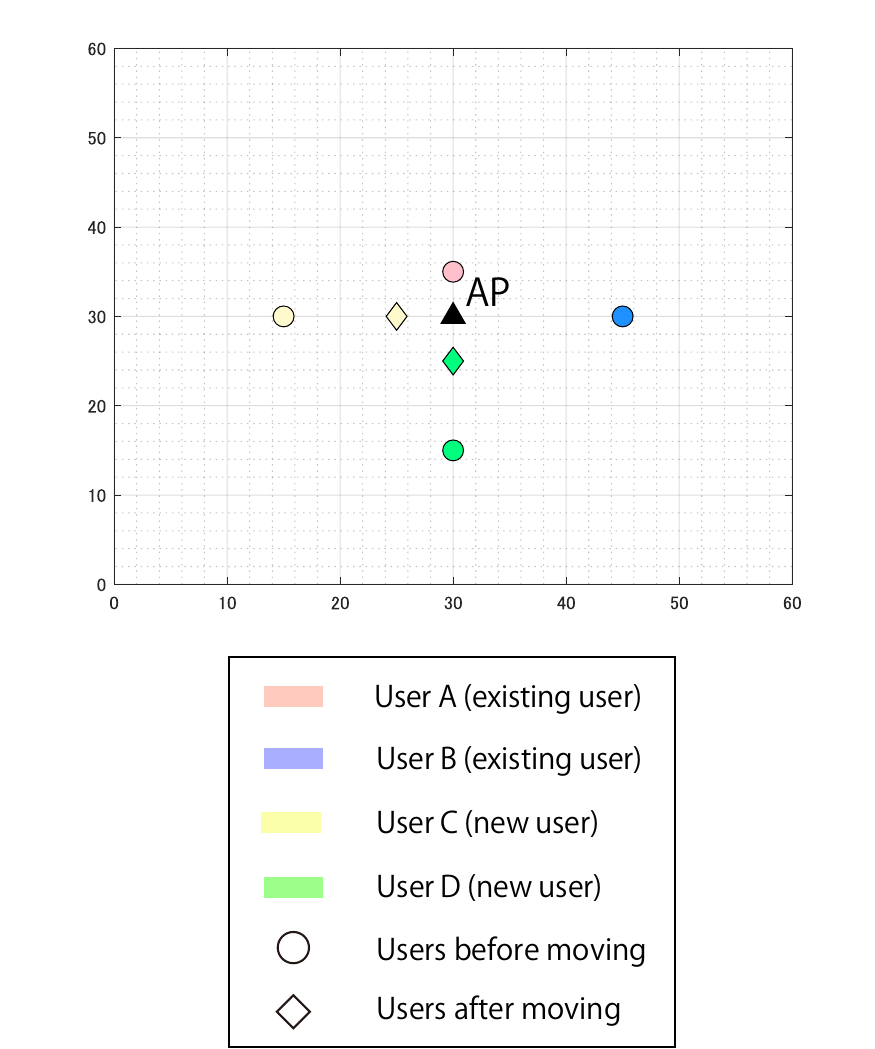}
    \subcaption{Optimal position of each user in conventional method \cite{miyata_CCNC}.}
  \end{minipage}
  \begin{minipage}[b]{0.48\linewidth}
    \centering
    \includegraphics[keepaspectratio, scale=0.3]{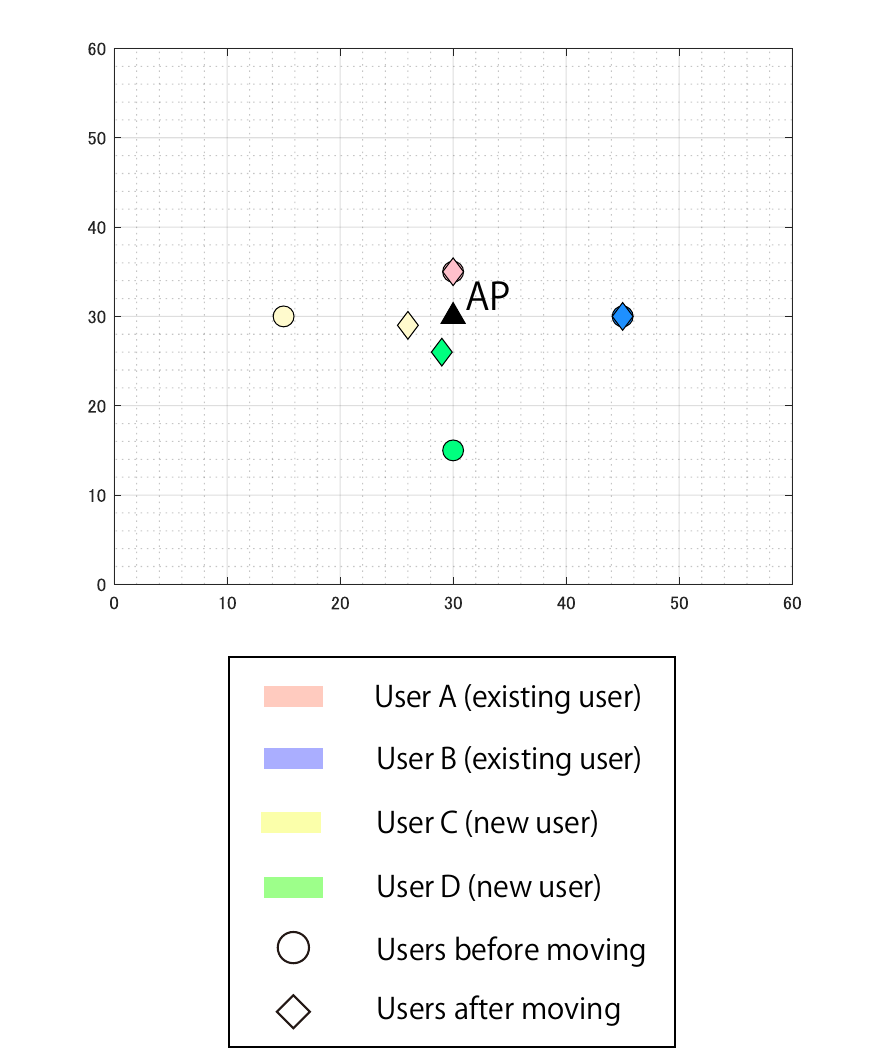}
    \subcaption{Optimal position of each user in proposed method.}
  \end{minipage}
  \caption{Optimal position of each user at \(\bm{d}_{A} = (5, 90^\circ)\) in pattern \((\mathrm{I}\hspace{-1.2pt}\mathrm{I})\).}
  \label{User_saiteki_position_pattern2_35}
\end{figure}

The parameter settings for our numerical analysis are shown in Table \ref{Para}. 

As mentioned in Section 4, we compared the system throughput as an evaluation method with the conventional methods \cite{miyata_CCNC}\cite{kato_CCNC}. Moreover, we use an improvement ratio of the system throughput \(\theta\) as an evaluation basis. Thus, the improvement ratio of the system throughput \(\Delta\theta\) is follows, 
\begin{equation}
\Delta\theta = \frac{\theta_{\mathrm{pro}}} {\theta_{\mathrm{non \, move}}}.
\label{equ_imp}
\end{equation}
Let \(\theta_{\mathrm{pro}}\) be the system throughput of the proposed method, \(\theta_{\mathrm{non move}}\) be the system throughput when users do not move.

In our numerical analysis, we analyze the characteristics of the users' optimal positions, and the system throughput by varying the position of users, as shown in Fig. \ref{Assumed_enviroment}(a)--(f). Here, each user is defined as user A, user B, user C, and user D. In the proposed method, there is no difference between these users. However, there is a clear difference between new users and existing users in the conventional methods \cite{miyata_CCNC}\cite{kato_CCNC}. Therefore, we define users A and B as existing users and users C and D as new users in order to compare the conventional methods \cite{miyata_CCNC}\cite{kato_CCNC}. As a specific user setup, the initial positions of all users except user A were fixed, and the system throughput was calculated when the initial position of user A \(\hat{\bm{d}_{A}}\) was changed. We assume 6 patterns of users' positions. In particular, patterns \((\mathrm{I})\)--\((\mathrm{I}\hspace{-1.2pt}\mathrm{I}\hspace{-1.2pt}\mathrm{I})\) is the case in which the distance between the AP and users is equal at the initial position of users. Furthermore, patterns \((\mathrm{I}\hspace{-1.2pt}\mathrm{V})\)--\((\mathrm{V}\hspace{-1.2pt}\mathrm{I})\) are cases in which all users cannot be equally distance from the AP.

In pattern \((\mathrm{I})\), we set the initial position of user B to \(\hat{\bm{d}}_{B} = (5, 0^\circ)\), user C to \(\hat{\bm{d}}_{C} = (5, 180^\circ)\), and user D to \(\hat{\bm{d}}_{D} = (5, -90^\circ)\), as shown in Fig. \ref{Assumed_enviroment}(a). In pattern \((\mathrm{I}\hspace{-1.2pt}\mathrm{I})\), we set the initial position of user B to \(\hat{\bm{d}}_{B} = (15, 0^\circ)\), user C to \(\hat{\bm{d}}_{C} = (15, 180^\circ)\), and user D to \(\hat{\bm{d}}_{D} = (15, -90^\circ)\), as shown in Fig. \ref{Assumed_enviroment}(b). In pattern \((\mathrm{I}\hspace{-1.2pt}\mathrm{I}\hspace{-1.2pt}\mathrm{I})\), we set the initial position of user B to \(\hat{\bm{d}}_{B} = (30, 0^\circ)\), user C to \(\hat{\bm{d}}_{C} = (30, 180^\circ)\), and user D to \(\hat{\bm{d}}_{D} = (30, -90^\circ)\), as shown in Fig. \ref{Assumed_enviroment}(c).

Fig. \ref{System_throughput_comparison_ver1} shows the system throughput of the proposed method and the other methods in the pattern \((\mathrm{I})\). In particular, Fig. \ref{System_throughput_comparison_ver1}(a) shows the improvement ratio of system throughput \(\Delta\theta\) between the proposed method and the case in which users do not move. Fig. \ref{System_throughput_comparison_ver1}(b) shows the system throughput of proposed and conventional methods \cite{miyata_CCNC}. In addition, Fig. \ref{System_throughput_comparison_ver1}(c) shows the system throughput of proposed and conventional methods \cite{kato_CCNC}. Similarly, Fig. \ref{System_throughput_comparison_ver2} shows the system throughput of the proposed method and the other methods in the pattern \((\mathrm{I}\hspace{-1.2pt}\mathrm{I})\). Furthermore, Fig. \ref{System_throughput_comparison_ver3} shows the system throughput of the proposed method and the other methods in the pattern \((\mathrm{I}\hspace{-1.2pt}\mathrm{I}\hspace{-1.2pt}\mathrm{I})\). The horizontal axis represents the initial position of user A.

From Figs. \ref{System_throughput_comparison_ver1}--\ref{System_throughput_comparison_ver3}, our proposed method improves system throughput compared to all conventional methods \cite{miyata_CCNC}\cite{kato_CCNC}. In particular, our proposed method improved system throughput by about \(6 \, \%\) at most compared to the conventional method \cite{miyata_CCNC}. Moreover, the improvement ratio is especially high at the position \(d_{A} = 30\). In order to explain the cause, we show Fig. \ref{User_saiteki_position_pattern2_60}. Fig. \ref{User_saiteki_position_pattern2_60} is the user's optimal position of each methods at \(\hat{\bm{d}}_{A} = (30, 90^\circ)\) in pattern \((\mathrm{I}\hspace{-1.2pt}\mathrm{I})\). 

The conventional method \cite{miyata_CCNC} is a method in which user C and user D move closer to the AP to obtain a higher transmission rate. However, the conventional method \cite{miyata_CCNC} doesn't consider the interference between users and the transmission rates obtained by other users. Therefore, users C and D move closer to the AP, and only user A connects to the AP at a far away position from the AP, as shown in Fig. \ref{User_saiteki_position_pattern2_60}(a). Thus, because the distance between user A and the AP is far away, the transmit power of user A \(P^{\mathrm{receive}}_{A}\) received by the AP is low. Moreover, compared to user A, the other users are connected relatively close to the AP. Therefore, the interference \(\Gamma_{A}\) received by user A also increases. As a result, only the transmission rate of user A will be extremely low, degrading the overall system performance.

Meanwhile, the optimal position of the users in the proposed method cannot be a situation which degrades the overall system performance. This is because users are moved at an equal distance from the AP, considering the interference between users, as shown in Fig. \ref{User_saiteki_position_pattern2_60}(b). Thus, we can assume that the system throughput of the proposed method is significantly higher than that of the conventional method. For these reasons, the system throughput of the proposed method improves very highly, especially in the pattern \((\mathrm{I})\) in which only user A is significantly far away from the AP.

\begin{figure*}[!t]
  \begin{minipage}[b]{0.33\linewidth}
    \centering
    \includegraphics[keepaspectratio, scale=0.35]{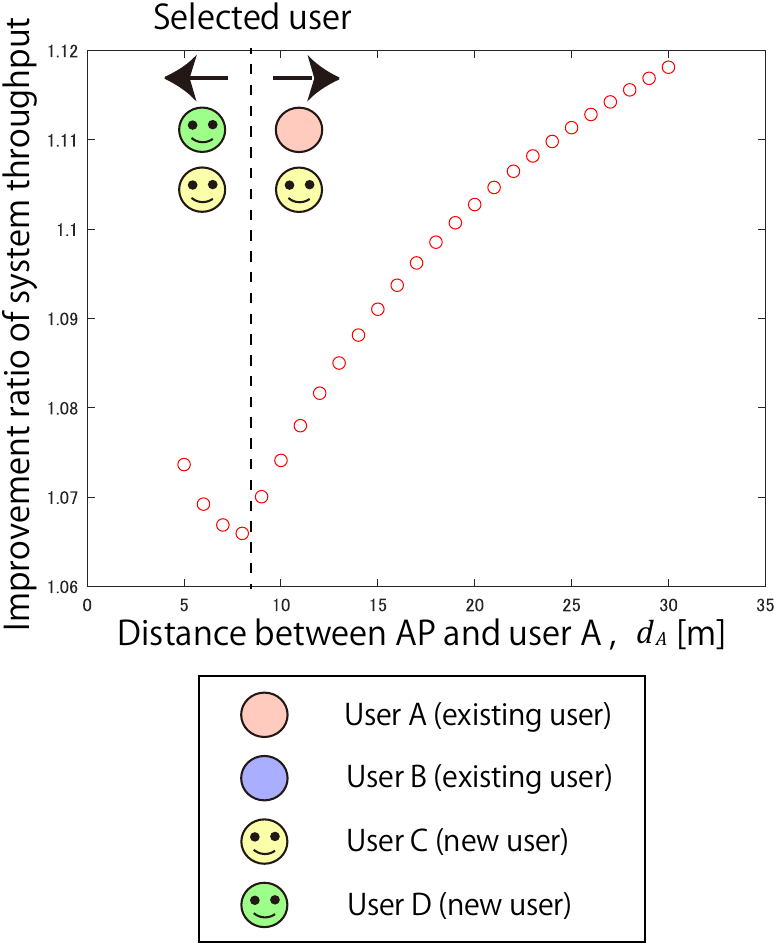}
    \subcaption{Improvement ratio of system throughput between the proposed method and the case in which users do not move.}
  \end{minipage}
  \begin{minipage}[b]{0.33\linewidth}
    \centering
    \includegraphics[keepaspectratio, scale=0.35]{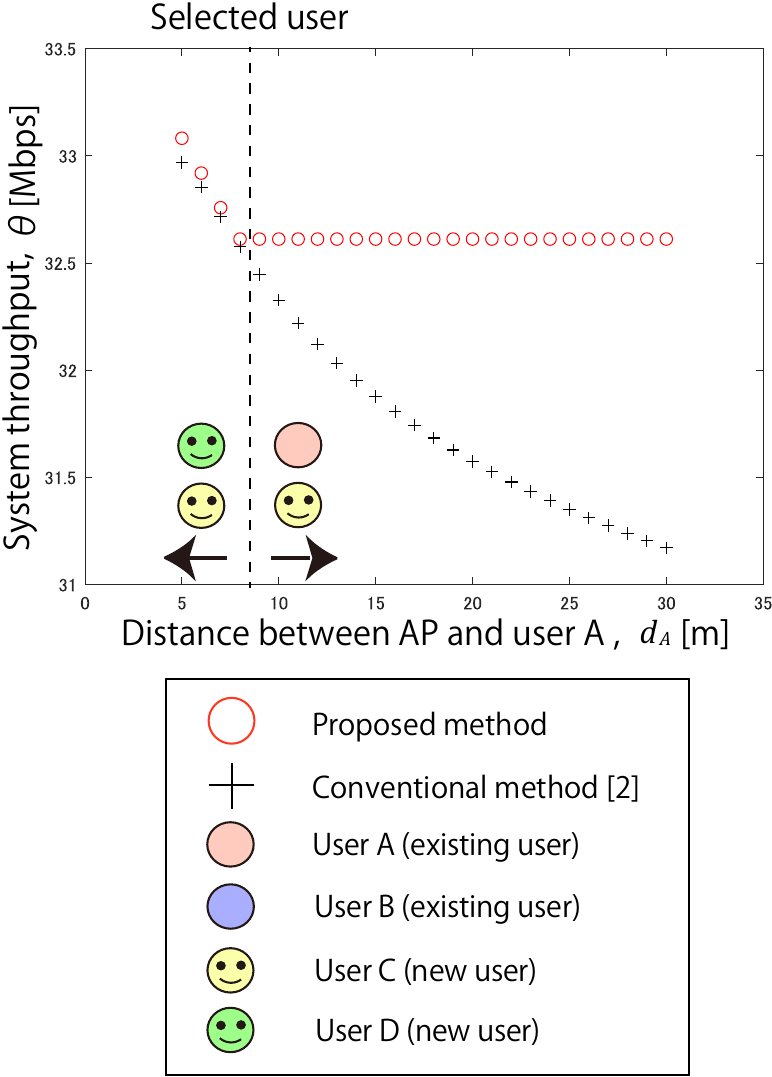}
    \subcaption{System throughput of proposed and conventional methods \cite{miyata_CCNC}.}
  \end{minipage}
  \begin{minipage}[b]{0.33\linewidth}
    \centering
    \includegraphics[keepaspectratio, scale=0.35]{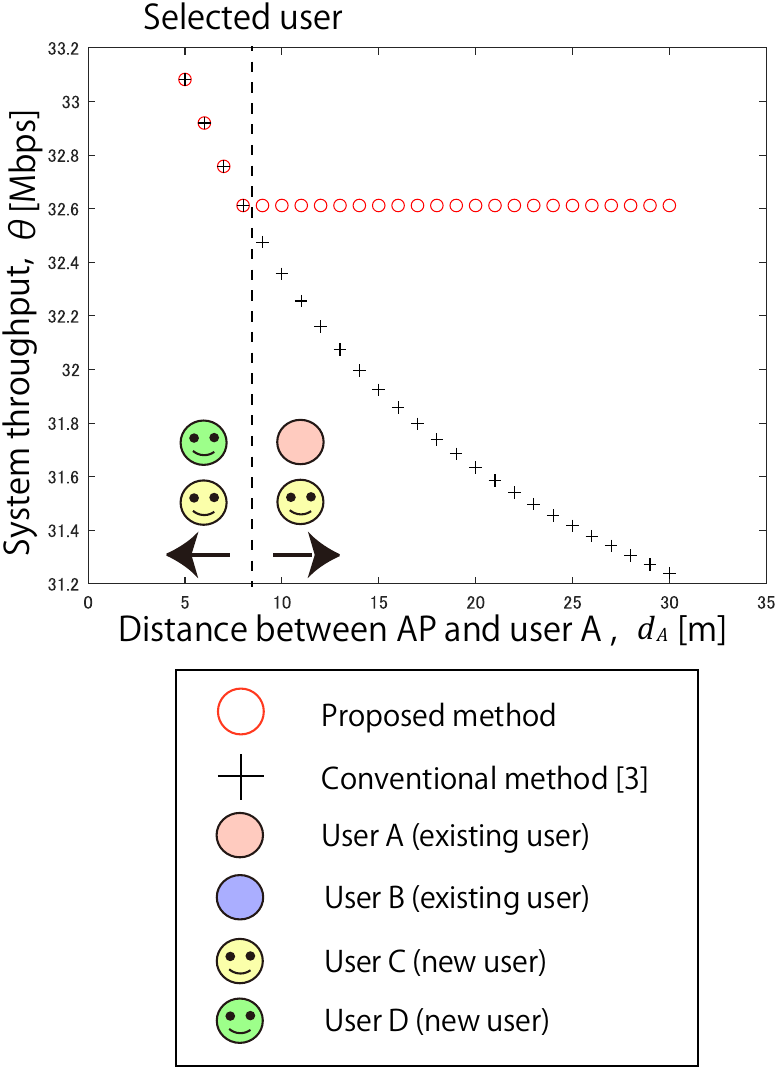}
    \subcaption{System throughput of proposed and conventional methods \cite{kato_CCNC}.}
  \end{minipage}
  \caption{System throughput of each method in pattern \((\mathrm{I}\hspace{-1.2pt}\mathrm{V})\).}
  \label{System_throughput_comparison_ver4}
\end{figure*}

\begin{figure*}[!t]
  \begin{minipage}[b]{0.33\linewidth}
    \centering
    \includegraphics[keepaspectratio, scale=0.35]{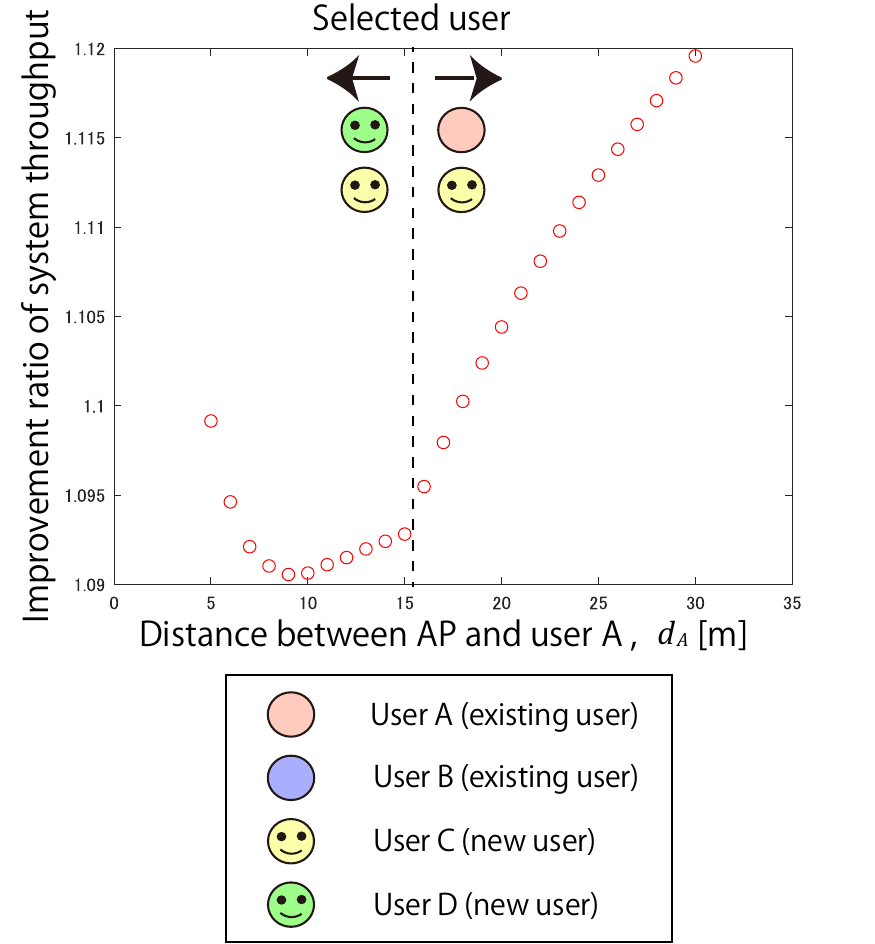}
    \subcaption{Improvement ratio of system throughput between the proposed method and the case in which users do not move.}
  \end{minipage}
  \begin{minipage}[b]{0.33\linewidth}
    \centering
    \includegraphics[keepaspectratio, scale=0.35]{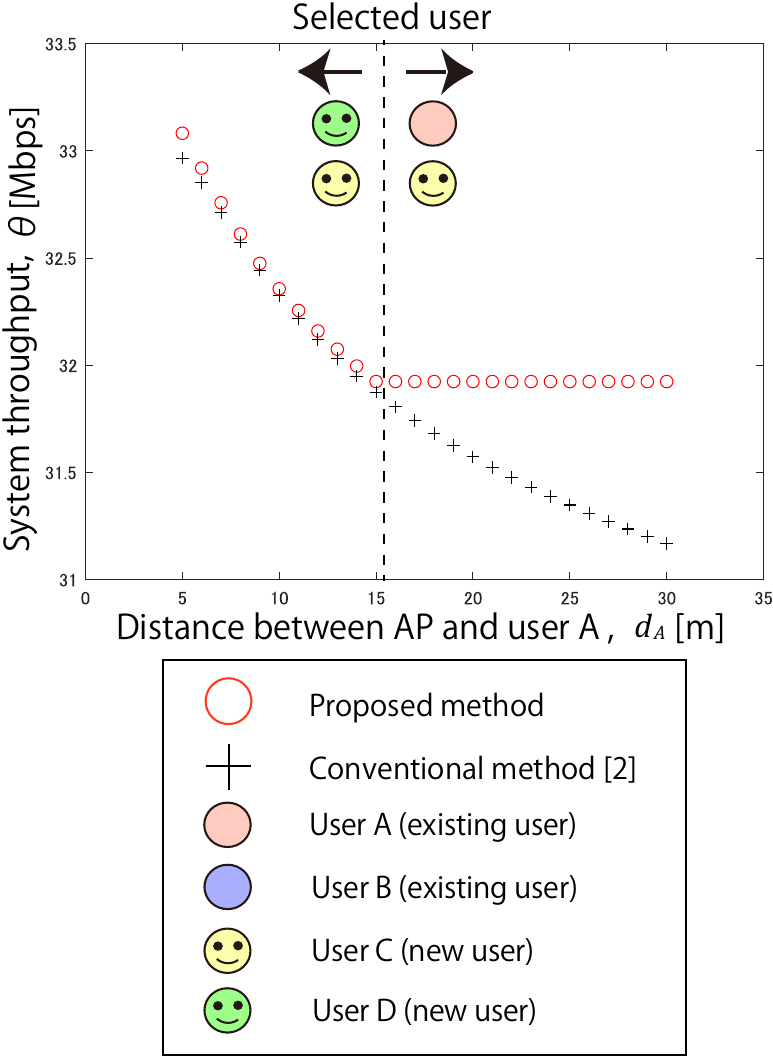}
    \subcaption{System throughput of proposed and conventional methods \cite{miyata_CCNC}.}
  \end{minipage}
  \begin{minipage}[b]{0.33\linewidth}
    \centering
    \includegraphics[keepaspectratio, scale=0.35]{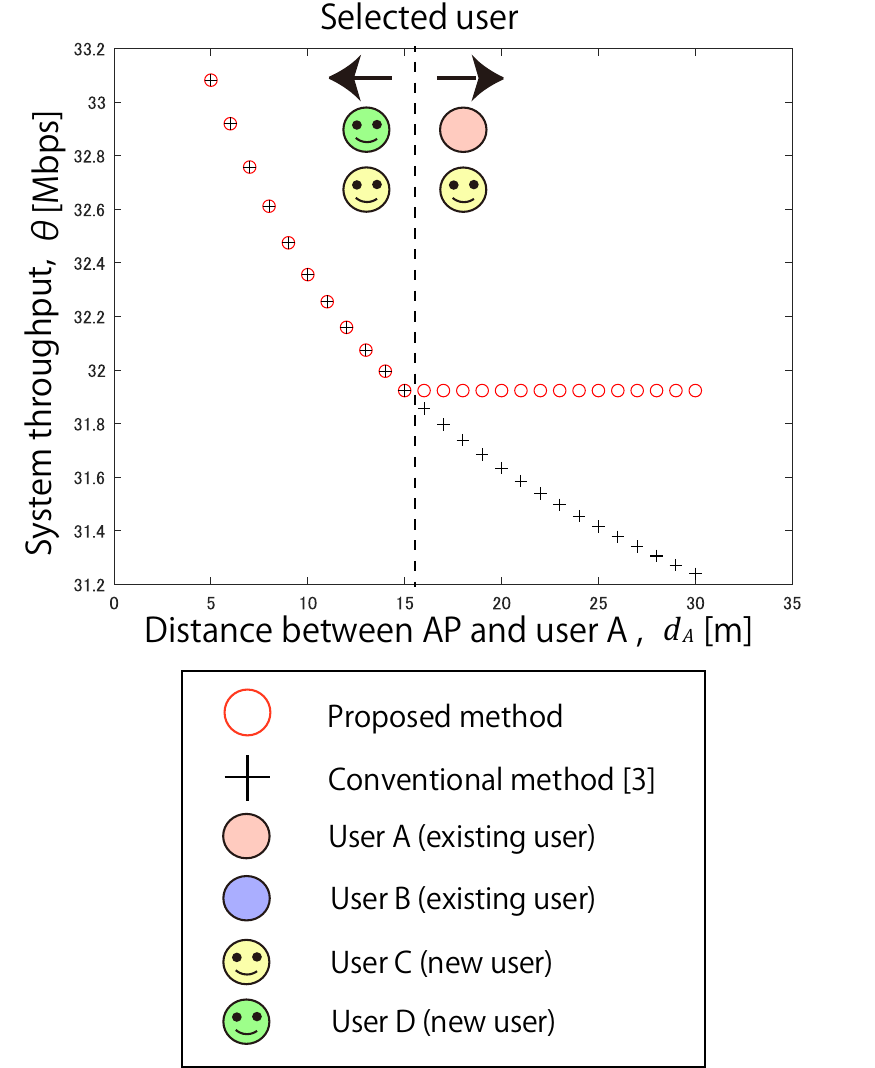}
    \subcaption{System throughput of proposed and conventional methods \cite{kato_CCNC}.}
  \end{minipage}
  \caption{System throughput of each method in pattern \((\mathrm{V})\).}
  \label{System_throughput_comparison_ver5}
\end{figure*}

\begin{figure*}[!t]
  \begin{minipage}[b]{0.33\linewidth}
    \centering
    \includegraphics[keepaspectratio, scale=0.35]{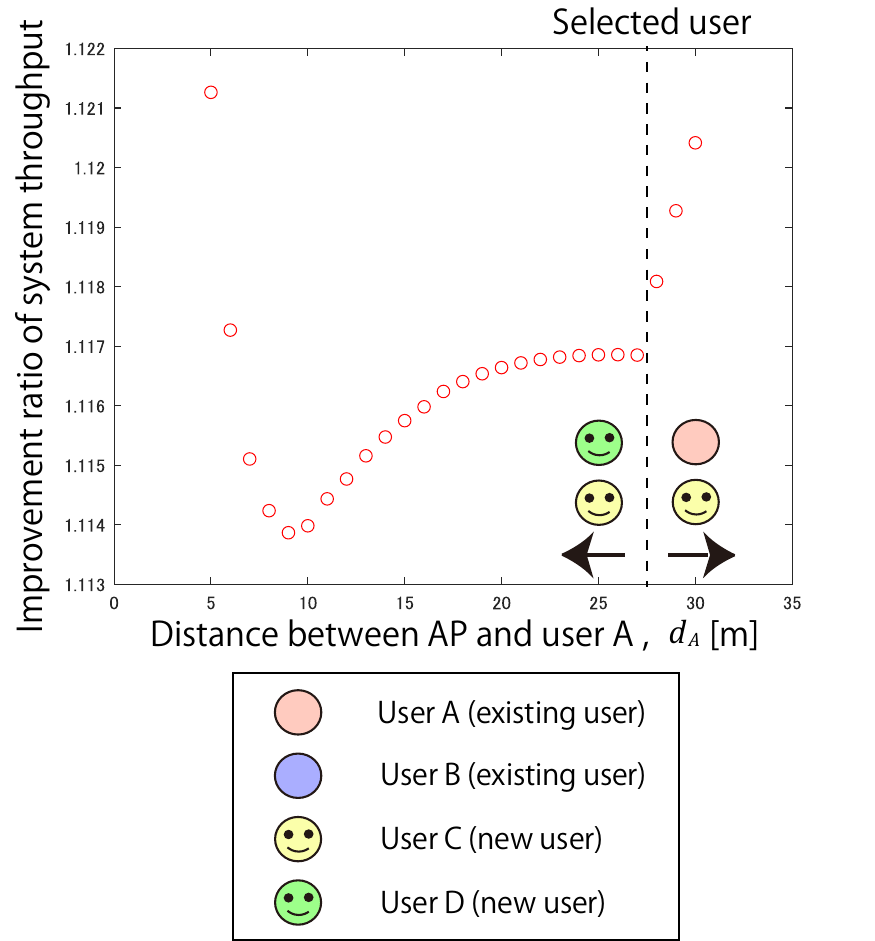}
    \subcaption{Improvement ratio of system throughput between the proposed method and the case in which users do not move.}
  \end{minipage}
  \begin{minipage}[b]{0.33\linewidth}
    \centering
    \includegraphics[keepaspectratio, scale=0.35]{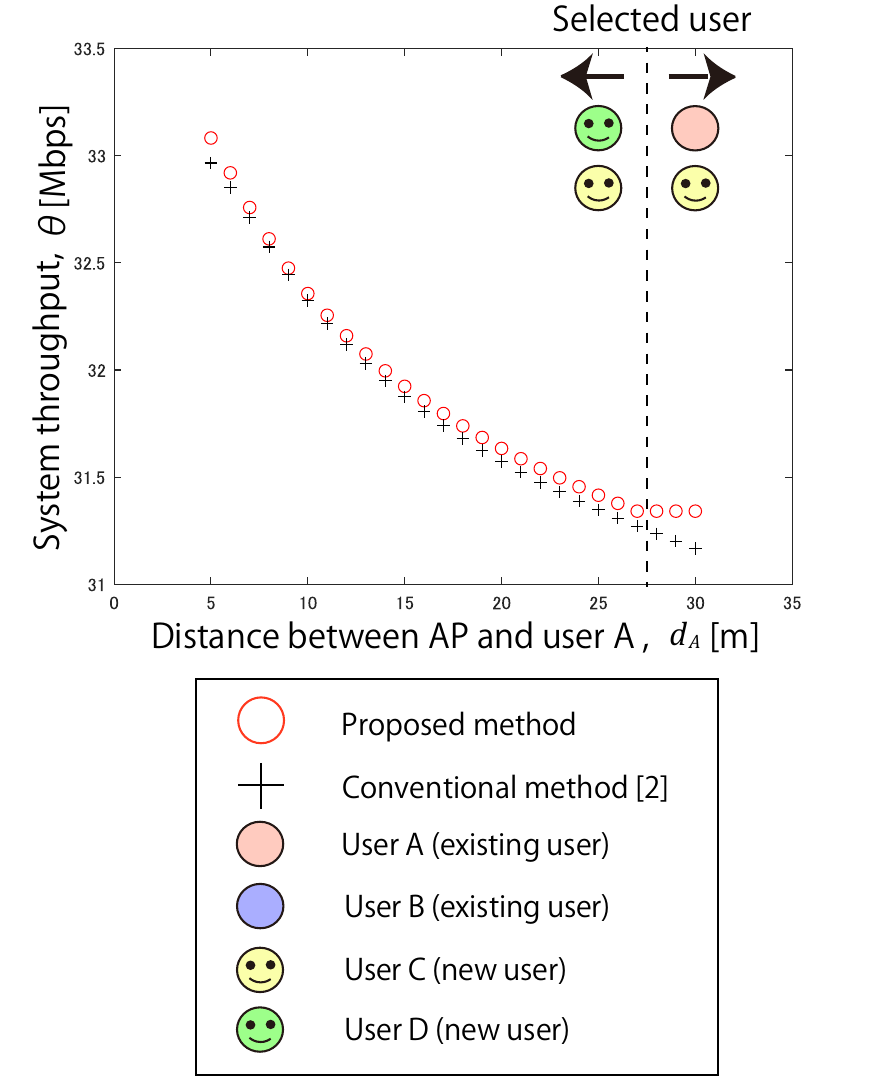}
    \subcaption{System throughput of proposed and conventional methods \cite{miyata_CCNC}.}
  \end{minipage}
  \begin{minipage}[b]{0.33\linewidth}
    \centering
    \includegraphics[keepaspectratio, scale=0.35]{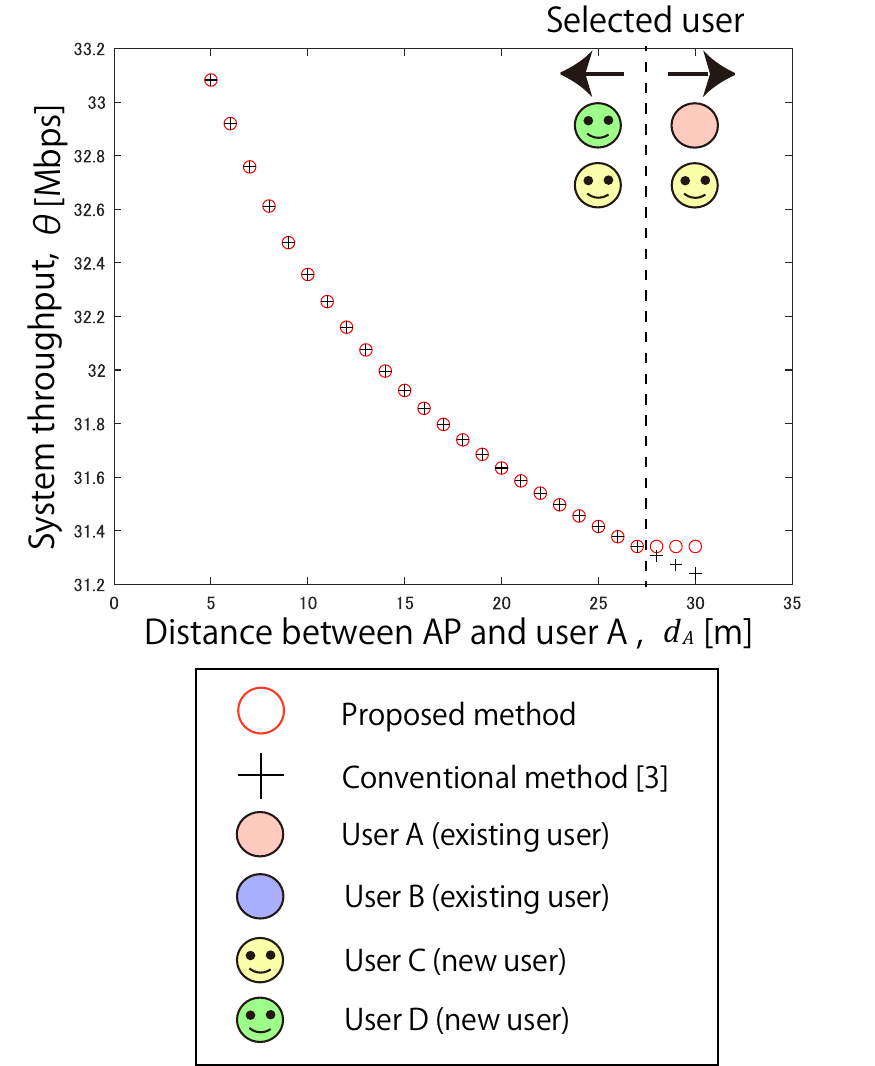}
    \subcaption{System throughput of proposed and conventional methods \cite{kato_CCNC}.}
  \end{minipage}
  \caption{System throughput of each method in pattern \((\mathrm{V}\hspace{-1.2pt}\mathrm{I})\).}
  \label{System_throughput_comparison_ver6}
\end{figure*}

\begin{figure}[!t]
  \begin{minipage}[b]{0.48\linewidth}
    \centering
    \includegraphics[keepaspectratio, scale=0.3]{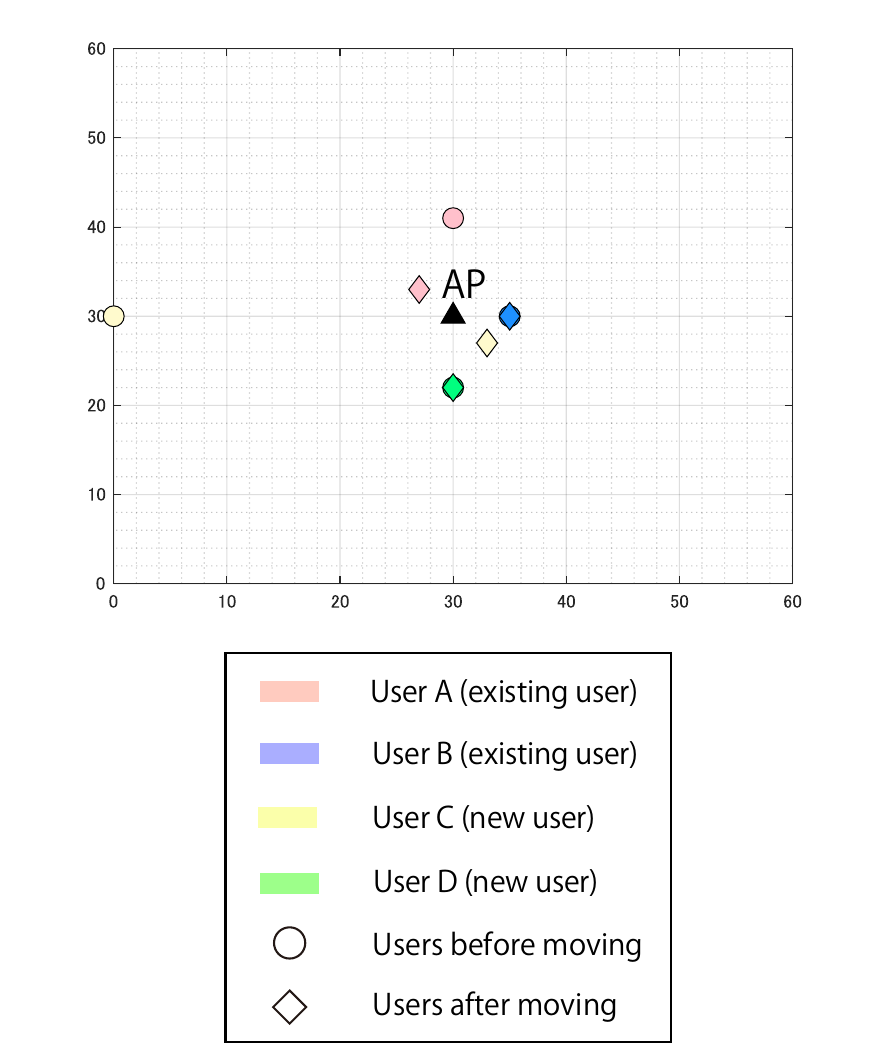}
    \subcaption{Optimal position of each user at \(\hat{\bm{d}}_{A} = (11, 90^\circ)\).}
  \end{minipage}
  \begin{minipage}[b]{0.48\linewidth}
    \centering
    \includegraphics[keepaspectratio, scale=0.3]{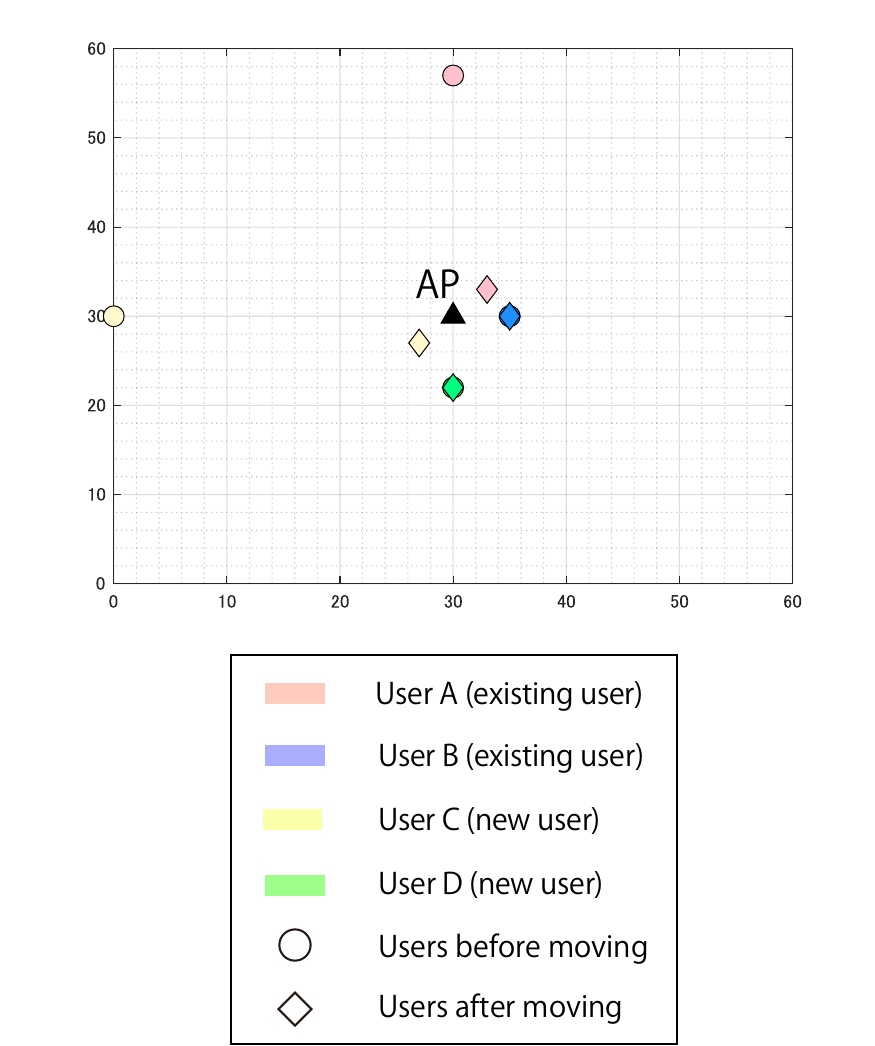}
    \subcaption{Optimal position of each user at \(\hat{\bm{d}}_{A} = (27, 90^\circ)\).}
  \end{minipage}
  \caption{Optimal position of each user for the proposed method in pattern \((\mathrm{I}\hspace{-1.2pt}\mathrm{V})\).}
  \label{User_saiteki_position_pattern4}
\end{figure}

However, for all patterns, the improvement ratio of the system throughput in the proposed method is not better than that of the conventional methods \cite{miyata_CCNC}\cite{kato_CCNC} when the initial position of user A is close to the AP. This cause can be explained using Fig. \ref{User_saiteki_position_pattern2_35}. Fig. \ref{User_saiteki_position_pattern2_35} shows the users' optimal positions for the proposed and the conventional method \cite{miyata_CCNC} when the initial position of user A in the pattern \((\mathrm{I}\hspace{-1.2pt}\mathrm{I})\) is \(\hat{\bm{d}}_{A} = (5, 90^\circ)\).

The closer the user is to the AP, the higher the transmission rate. However, the closer the distance between the user and the AP, the stronger the effect of interference on other users. Thus, the transmission rate obtained by other users is reduced. The interference has a greater effect on users who are far away from the AP. However, if the initial position of user A in the pattern \((\mathrm{I}\hspace{-1.2pt}\mathrm{I})\) is \(\hat{\bm{d}}_{A} = (5, 90^\circ)\), there is no user who gets an extremely low transmission rate. Thus, because there are no users who suffer large amounts of interference, the optimal position of the user is close to the AP, even in the proposed method considering the interference. As a result, the optimal position of users does not differ largely compared to the conventional method \cite{miyata_CCNC} in which users are connected in close to the AP. Therefore, the improvement ratio of the system throughput is not good.

From Fig. \ref{System_throughput_comparison_ver3}, we can see that there is no difference in value between the system throughput of the proposed and the conventional method \cite{kato_CCNC} in the pattern \((\mathrm{I}\hspace{-1.2pt}\mathrm{I}\hspace{-1.2pt}\mathrm{I})\). This indicates that the user's optimal position of the proposed and the conventional methods \cite{kato_CCNC} are exactly the same. This is because the transmission rates of users A and B, who were already connected to the AP, are not extremely low. In the pattern \((\mathrm{I})\) and the pattern \((\mathrm{I}\hspace{-1.2pt}\mathrm{I})\), from a certain initial position of user A, only user A may connect more far away from the AP than the other users. Thus, because user A gets an extremely low transmission rate, the system throughput will be low unless the user A is moved. Therefore, the proposed method, which considers the movement of existing users, is more effective than the conventional method \cite{kato_CCNC}, which only considers new users. However, in the pattern \((\mathrm{I}\hspace{-1.2pt}\mathrm{I}\hspace{-1.2pt}\mathrm{I})\), there is no case in which only one user connects at an extremely remote position from the AP even if the initial position of user A \(\hat{\bm{d}}_{A}\) changes, because the initial positions of users other than user A are already far away from the AP. In other words, there are no existing users (user A and user B) who obtain extremely low transmission rates. As a result, the system throughput of the conventional and proposed methods did not differ because only new users (users C and D) move in order to improve the system throughput.

\begin{figure}[!t]
  \begin{minipage}[b]{0.48\linewidth}
    \centering
    \includegraphics[keepaspectratio, scale=0.3]{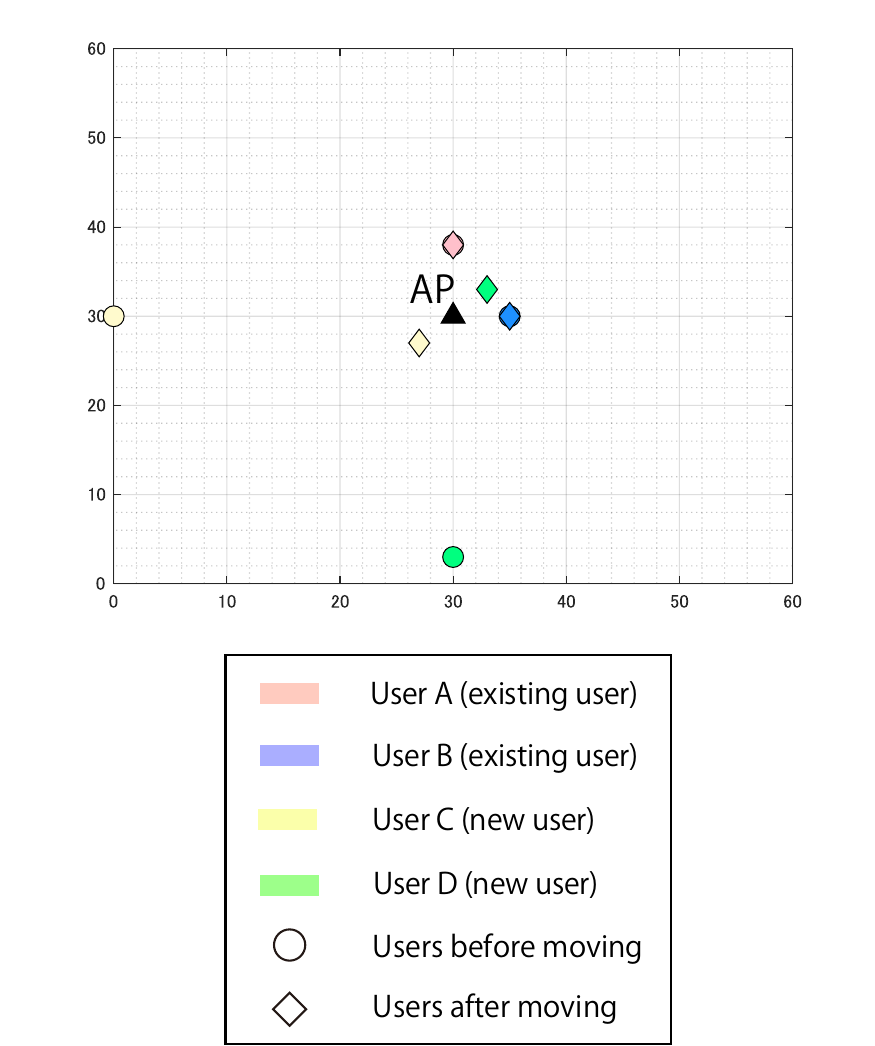}
    \subcaption{Optimal position of each user at \(\hat{\bm{d}}_{A} = (8, 90^\circ)\).}
  \end{minipage}
  \begin{minipage}[b]{0.48\linewidth}
    \centering
    \includegraphics[keepaspectratio, scale=0.3]{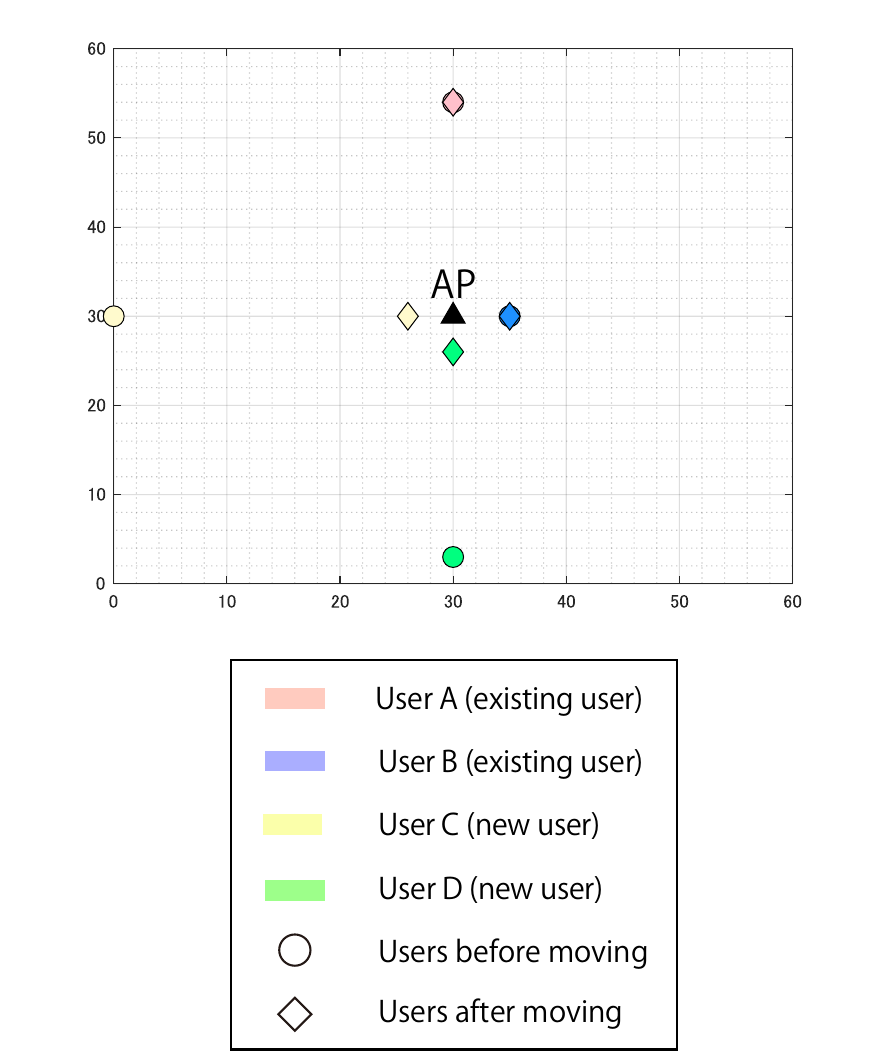}
    \subcaption{Optimal position of each user at \(\hat{\bm{d}}_{A} = (24, 90^\circ)\).}
  \end{minipage}
  \caption{Optimal position of each user for the proposed method in pattern \((\mathrm{V}\hspace{-1.2pt}\mathrm{I})\).}
  \label{User_saiteki_position_pattern6}
\end{figure}

Here, Fig. \ref{System_throughput_comparison_ver1}(a), \ref{System_throughput_comparison_ver2}(a), \ref{System_throughput_comparison_ver3}(a) show the improvement ratio of the system throughput \(\Delta\theta\) between the proposed method and the case in which users do not move. From these figures, we can assume that the initial position where all users are at equal distance from the AP has the lowest improvement ratio of the system throughput. However, there are very few situations in which users are connected at equal distance from the AP in a real environment. Therefore, we next analyze the system throughput and users' positions in a situation where all users cannot be at an equal distance from the AP. Specifically, we analyze the patterns in Fig. \ref{Assumed_enviroment}(d)--(f).

Fig. \ref{System_throughput_comparison_ver4}, \ref{System_throughput_comparison_ver5}, \ref{System_throughput_comparison_ver6} show the system throughput characteristics for the pattern \((\mathrm{I}\hspace{-1.2pt}\mathrm{V})\), pattern \((\mathrm{V})\), pattern \((\mathrm{V}\hspace{-1.2pt}\mathrm{I})\). From Fig. \ref{System_throughput_comparison_ver4}--\ref{System_throughput_comparison_ver6}, even if the initial positions of each user are unequal, the system throughput improves when the initial position of user A is far away from the AP. In particular, the closer the initial position of user D is to the AP, the higher the system throughput, such as in the pattern \((\mathrm{I}\hspace{-1.2pt}\mathrm{V})\). This is similar to the situation in Fig. \ref{User_saiteki_position_pattern2_60}. In the conventional methods \cite{miyata_CCNC}\cite{kato_CCNC}, when the initial position of user A \(\hat{\bm{d}}_{A}\) is far from the AP, the transmission rate obtained by user A is extremely low. As a result, the system throughput is low. Therefore, in the pattern \((\mathrm{I}\hspace{-1.2pt}\mathrm{V})\), the closer the initial position of user D is to the AP, the higher the interference to user A is. The system throughput in the pattern \((\mathrm{I}\hspace{-1.2pt}\mathrm{V})\) is greatly improved because the transmission rate of user A is greatly reduced.

Next, we analyze the characteristics of the users' optimal positions when the users are at unequal distance. In order to analyze this characteristic, as an example, Fig. \ref{User_saiteki_position_pattern4} shows the optimal position of user A in the pattern \((\mathrm{I}\hspace{-1.2pt}\mathrm{V})\) when the initial position of user A is \(\hat{\bm{d}}_{A} = (11, 90^\circ)\) and \(\hat{\bm{d}}_{A} = (27, 90^\circ)\). In addition, Fig. \ref{User_saiteki_position_pattern6} shows the optimal position of user A in the pattern \((\mathrm{V}\hspace{-1.2pt}\mathrm{I})\) when the initial position of user A is \(\hat{\bm{d}}_{A} = (8, 90^\circ)\) and \(\hat{\bm{d}}_{A} = (24, 135^\circ)\).

From Fig. \ref{User_saiteki_position_pattern4}, we can see that if only one of the user's initial positions is far from the AP (a situation that degrades the performance of the entire system due to one user), this user moves closer to the AP. Moreover, even if two users are far away from the AP in their initial positions (a situation that would not degrade the performance of the entire system), these users move closer to the AP. This is because of the following reasons. If there are users connected in closer to the AP, the transmission power of the users connected in closer will be higher. In addition, the transmission rate will be higher. However, the transmission rate of a user who is connected far away from the AP will be extremely lower, because the interference power will be very high compared to his own transmission power. Therefore, if there are users connected in closer to the AP, there will be users with extremely low transmission rates due to interference, and system throughput will be degraded. From this, users who are far away from the AP need to increase their own transmit power and reduce the interference. As a result, the user who is initially far away from the AP moves closer to the AP. For the same reason, users who are connected at far away from the AP move closer to the AP in Fig. \ref{User_saiteki_position_pattern6}(a). From Fig. \ref{User_saiteki_position_pattern6}(b), even when only one of the users' initial positions is connected closer to the AP, users who are connected far away from the AP will move closer to the AP. In this case, only user A is connected far away from the AP, which causes the transmission rate of user A to be extremely low. However, even if one user gets an extremely low transmission rate, three users make up for the negative factor, increasing system throughput, because the three users connect in closer to the AP. 

In exceptional cases, there are cases in which it is better to place all users at an equal distance. This is the case for the pattern \((\mathrm{I}\hspace{-1.2pt}\mathrm{I})\) and pattern \((\mathrm{I}\hspace{-1.2pt}\mathrm{I}\hspace{-1.2pt}\mathrm{I})\). This is because only two users cannot make up for the lower transmission rate, even if the two users who are far away move closer to the AP. As a result, we can assume that the optimal position would be to place all users at an equal distance.

From the above discussion, our conclusion is two factors:
\begin{quote}
\begin{itemize}
\item Basically, a user who is far away from the AP should be moved closer to the AP.
\item If there are no users connected closer to the AP, position of all users equally far away because only two users cannot make up for the lower transmission rate.
\end{itemize}
\end{quote}

\section{CONCLUSION}
In this paper, we propose a AP connection method to maximize the system throughput considering interference frequency and initial position of all users. The proposed method successfully maximizes system throughput considering the interference frequency and the performance degradation of the overall system due to one user obtaining an extremely low transmission rate. By proving that this problem is based on a potential game, we were able to theoretically analyze the complex wireless network environment. In the future, we will consider user incentives and analyze a situation when the number of users increases.

\section*{ACKNOWLEDGEMENT}
These research results were obtained from the commissioned research (No.05601) by National Institute of Information and Communications Technology (NICT), Japan. In addition, this work was supported by JSPS KAKENHI Grant Numbers JP19K11947, JP22K12015, JP20H00592, JP21H03424.

\bibliography{refs_kato}
\bibliographystyle{IEEEtran}

\begin{IEEEbiography}[{\includegraphics[width=1in,height=1.25in,clip,keepaspectratio]{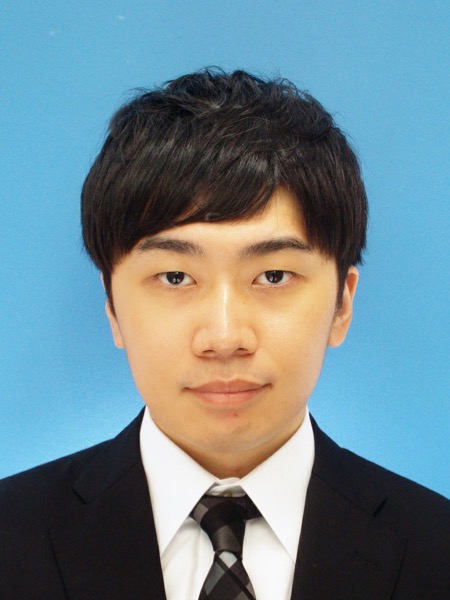}}]{Yu Kato} was born in Shizuoka, JPN.
He received the B.E. degree from the Shibaura Institute of Technology in 2022.
He is currently undertaking a M.E. candidate of the Department of Electrical Engineering and Computer Science in the Shibaura Institute of Technology in 2023.
His research interests include mathematical modelling and analysis for AP connection method considering user behaviour, and game theory in communication networks.
\end{IEEEbiography}

\begin{IEEEbiography}[{\includegraphics[width=1in,height=1.25in,clip,keepaspectratio]{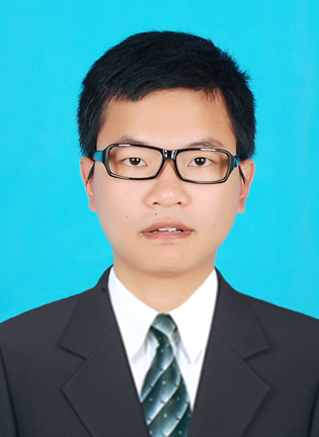}}]{Jiquan Xie} received his B.S. and M.S. degrees from the University of Electronic Science and Technology of China (UESTC), Chengdu, China, in 2013 and 2016, respectively. From 2018 to 2019, he was a Research Assistant in Shanghai Jiao Tong University, Shanghai, China. He received his Ph.D. degree from Nagoya University in 2021. He is currently a collaborative researcher at Nagoya University, Japan.

Dr. Xie was awarded the IEEE Nagoya section Young Researcher Award in 2021, and the Excellent Doctoral Award of Informatics Department at Nagoya University in 2022.
\end{IEEEbiography}

\begin{IEEEbiography}[{\includegraphics[width=1in,height=1.25in,clip,keepaspectratio]{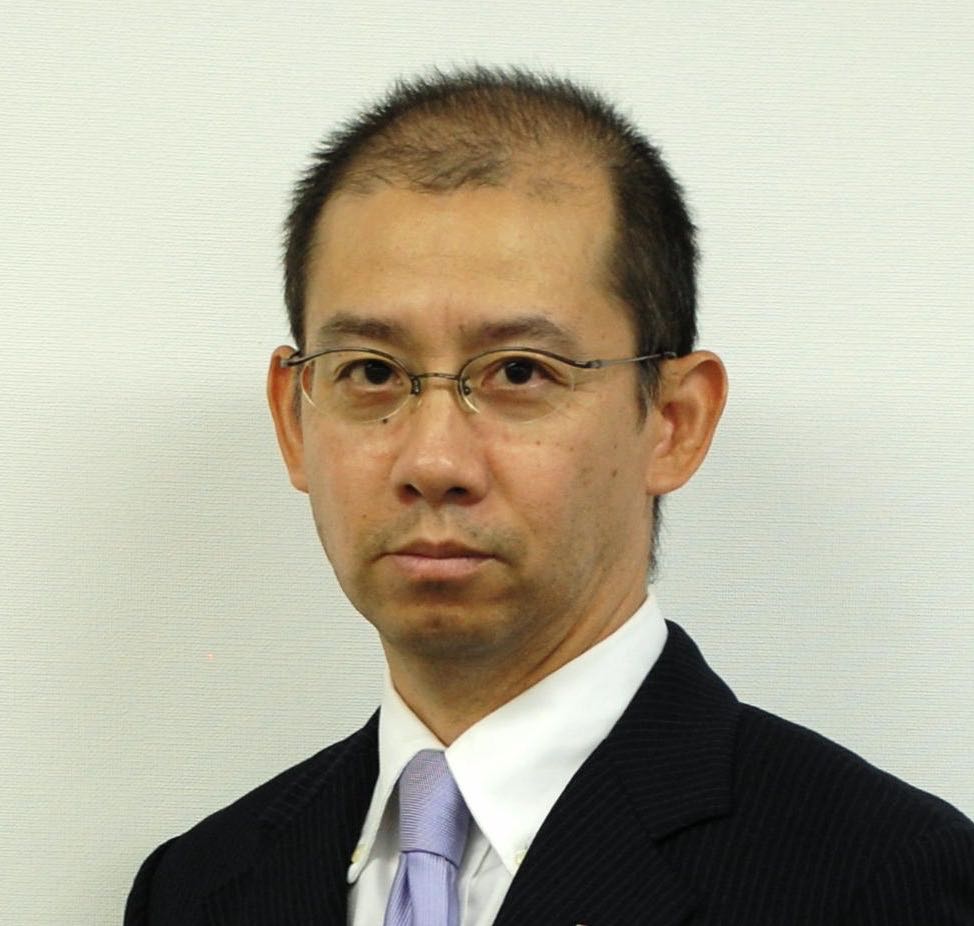}}]{Tutomu Murase} was born in Kyoto, Japan in 1961.
He received his M.E. degree from Graduate School of Engineering Science, Osaka University, Japan, in 1986. He also received his PhD degree from Graduate School of Information Science and Technology, Osaka University in 2004. He joined NEC Corporation Japan in 1986. He was a visiting professor in Tokyo Institute of Technology in 2012—2014. He is currently a professor in Nagoya University, Japan.
He has been engaged in researches on traffic management for high-quality and high-speed internet. His current interests include transport and session layer traffic control, wireless network resource management and network security. He is also interested in user cooperative mobility research. He received Best Tutorial Paper Award on his invited paper about QoS control for overlay networks in IEICE transaction on communication in 2006. He has been served as TPC for many IEEE conferences and workshops. He has more than 90 registered patents including some international patents. He was a secretary of IEEE Communications Society Japan Chapter. He is a member of IEEE and a fellow of IEICE.
\end{IEEEbiography}

\begin{IEEEbiography}[{\includegraphics[width=1in,height=1.25in,clip,keepaspectratio]{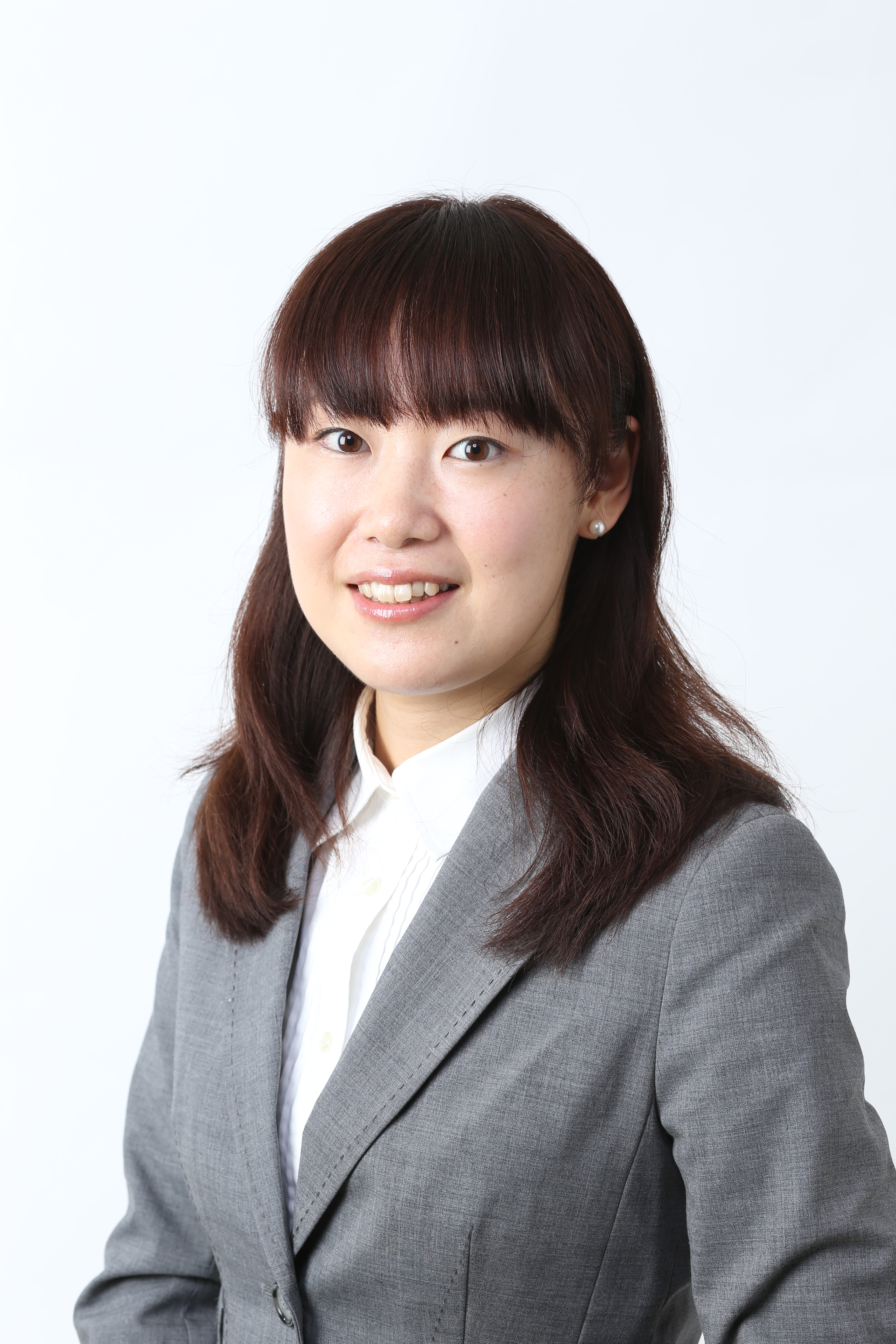}}]{Sumiko Miyata} received the B.E. degree from the Shibaura Institute of
Technology in 2007, and the M.E. and D.E. degrees from the Tokyo
Institute of Technology in 2009 and 2012 respectively. From 2012 to
2015, she was a research associate at the Kanagawa University. Since
2015, she had been an assistant professor at the Shibaura Institute of
Technology. Since 2018, she has been an associate professor at the
Shibaura Institute of Technology. Her research interests include
mathematical modelling and analysis for QoS performance evaluation,
queuing theory, game theory and resource allocation problems in
communication networks and information security.
\end{IEEEbiography}

\end{document}